\pdfoutput=1
\documentclass[sigconf, nonacm]{acmart}
\usepackage{amsmath,amssymb,amsthm,mathtools}
\usepackage{xcolor}
\usepackage{hyperref}
\usepackage{makecell}
\usepackage{booktabs}
\usepackage{threeparttable}

\usepackage{subcaption}

\usepackage[skip=5pt]{caption} % Adjust spacing below captions
\usepackage{balance}
\usepackage{tablefootnote}
\usepackage{multirow}
\usepackage{enumitem}

\newcommand{\Paragraph}[1]{\par\addvspace{2pt}\noindent\textbf{#1.}}

\theoremstyle{plain}
\newtheorem{theorem}{Theorem}

\newtheorem{corollary}[theorem]{Corollary}

\theoremstyle{remark}

\definecolor{okBG}{HTML}{009E73}
\definecolor{okV}{HTML}{D55E00}
\definecolor{okB}{HTML}{0072B2}
\definecolor{okSB}{HTML}{56B4E9}

\newcounter{discussion}
\newcommand{\DiscussionPara}[2]{%
  \refstepcounter{discussion}\label{#1}%
  \noindent\textbf{\textsc{\textit{Discussion}~\thediscussion}: #2.}\hspace{0.55em}\ignorespaces
}

\begin{document}
    \title{
    A Power Law in Logarithm's Clothing: On the Scalability of Graph-Based Vector Search}

	\author{Sajad Faghfoor Maghrebi}
	\orcid{0009-0009-0882-4999}
	\affiliation{%
		  \institution{University of Toronto}
		  \country{Canada}
		}
	\email{smaghrebi@cs.toronto.edu}

    \author{Navid Eslami}
	\orcid{0009-0001-3522-5846}
	\affiliation{%
		  \institution{University of Toronto}
		  \country{Canada}
		}
	\email{navideslami@cs.toronto.edu}
	
	\author{Niv Dayan}
	\orcid{0000-0003-0314-0167}
	\affiliation{%
		  \institution{University of Toronto}
		  \country{Canada}
		}
	\email{nivdayan@cs.toronto.edu}

\begin{abstract}
Most vector databases rely on graph-based indexes, notably HNSW and Vamana, for approximate nearest neighbor search. With the recent widespread adoption of embedding models, the datasets these databases store have grown rapidly. This growth raises a natural question: at a fixed accuracy, how does search cost scale with dataset size? The prevailing answer is that search cost grows poly-logarithmically. Yet this claim is formally proven only under special conditions; for the indexes used in practice, it is asserted without proof. The claim is also largely untested: standard benchmarks measure search cost at a single dataset size, not across~sizes.

We put this claim to the test and measure how search cost scales across dataset sizes. The answer turns out to depend on the scale itself. (1)~As long as the dataset size~$N$ is small relative to the data's intrinsic dimensionality, search cost grows in step with $N^{c}$ for a constant $0 < c < 1$. We call this scaling a \emph{Sublinear Power Law}. (2)~Once $N$ is large enough, the growth slows to subpolynomial, consistent with the prevailing poly-logarithmic claim. Empirically, the Sublinear Power Law appears on every dataset, mostly up to its full size, and at every recall target, query hardness level, and index configuration we test. Still, the transition to subpolynomial growth does appear on the two datasets that grow large enough relative to their intrinsic dimensionality. Underlying both behaviors is a single mechanism: the intrinsic dimensionality a dataset exhibits grows with its size, until the data resolves its underlying distribution. Higher intrinsic dimensionality packs more vectors into the small neighborhood around the query that the search must examine.

We present a unifying theory that formalizes beam-search cost and explains all of our observations. For both exact and bounded-degree constructions, we prove the Sublinear Power Law and the eventual transition to poly-logarithmic scaling, and derive the scale at which the transition occurs. We also develop models that predict the power-law exponents for any recall target and index configuration. These models provide a principled way to navigate trade-offs among search cost, insertion cost, and recall as data grows.
\end{abstract}
	
	\maketitle

\section{Introduction}
\label{sec:introduct}

Unstructured data, spanning documents, images, videos, and audio, has grown rapidly~\cite{analyticDB-V}. Querying such datasets using traditional approaches based on attribute values and keywords fails to capture semantic similarity~\cite{wang2021milvus}. Meanwhile, advances in learning-based methods, especially deep neural networks, have enabled models that transform unstructured data into dense vector representations suitable for semantic search~\cite{reimers2019sentencebertsentenceembeddingsusing}. These embedding models are trained so that semantic similarity corresponds to geometric proximity: the vectors closest to a query usually represent the most relevant data items.

\Paragraph{Vector Search}
Together, these trends have driven the rise of vector databases, which index the embeddings produced by these models~\cite{wang2021milvus, baevski2020wav2vec20frameworkselfsupervised}. Their core retrieval primitive is nearest neighbor search. However, computing the exact nearest neighbors is prohibitively expensive at scale~\cite{whynns, indyk1998approximateLSH, borodin1999lower}. Vector databases instead accept a small loss in accuracy and perform approximate nearest neighbor (ANN) search. To do so, they use an index that narrows the search to a small set of promising~candidates~\cite{NSG, DiskANN, HNSW}.

\Paragraph{Graph-based Indexes}
The ANN indexes that perform best in practice are graph-based~\cite{wang2021comprehensive, pmlr-v235-lu24l}. A graph-based index represents each vector as a node, connected to a subset of its nearby vectors, with a few longer-range edges that keep distant regions reachable in a few hops. To answer a query, the index runs beam search, a generalization of greedy search that walks the graph toward the query and returns the closest vectors encountered~\cite{HNSW, DiskANN}. Some graph constructions provide guarantees on accuracy and on reaching a query's neighborhood, but are often too expensive to build exactly~\cite{SNG, DiskANN}. Practical graph-based indexes draw on these exact constructions but build them with cheaper heuristics~\cite{HNSW, DiskANN}. The two most widely used such indexes, HNSW~\cite{HNSW} and Vamana~\cite{DiskANN}, power vector search across databases~\cite{milvus-index, weaviate-vectorindex, qdrant, cosmosdb-diskann, pgvector2023, pgvectorscale}, search engines~\cite{elastic_hnsw_merge_2025, opensearch-knn-tuning}, and~libraries~\cite{douze2025faisslibrary, nmslib, diskann-lib}.

\Paragraph{Sustaining Accuracy As Data Grows}
The data stored in vector databases grows continuously~\cite{singh2021freshdiskann, xu2023spfresh, baranchuk2023dedrift}. As the data grows, however, search accuracy---typically measured by \emph{recall}, the fraction of true nearest neighbors returned by the search---deteriorates~\cite{chatzakis2025darth, duan2025pgtuner, turbopuffer2024recall}. Maintaining stable recall therefore requires adjusting the index parameters in response~\cite{duan2025pgtuner, manohar2024parlayann, zhao2026grab}. A common practice is to periodically measure recall using exact search~\cite{qdrant_recall, vespa_recall, elastic_recall} and adjust parameters accordingly~\cite{manohar2024parlayann, turbopuffer2024recall, chatzakis2025darth}. If the measured recall falls short of the user's recall target, the user raises the search effort, remeasures, and repeats until the target is reached. Recent work in academia as well as in industry~\cite{zilliz_autoindex, turbopuffer2024recall} improves upon this practice by using statistical or learned models of the search effort required to achieve a recall target~\cite{chatzakis2025darth, zhang2026distribution, earlyterm2020}. Underlying all these approaches is a natural research question: \emph{At a fixed level of accuracy, how does search cost scale as the dataset grows?}

\Paragraph{The Logarithmic Folklore}
The literature answers this question with poly-logarithmic scaling. For exact graph constructions, the claim is sometimes backed by proof~\cite{SNG, indyk2023worst}. For the heuristic indexes used in practice, it is asserted and left unproven. In particular, the paper introducing HNSW claims logarithmic search cost, but the scaling experiments supporting this claim use only low-dimensional synthetic vectors ($d \le 8$), and its largest experiment deviates from logarithmic scaling~\cite{HNSW}. The paper introducing Vamana argues that greedy search on its exact construction takes logarithmically many steps, but reports no measurement of search cost across dataset sizes~\cite{DiskANN}. The claim has nonetheless propagated into industry documentation~\cite{douze2025faisslibrary, milvus-index, opensearch-knn-tuning, pinecone_hnsw, redis-hnsw-blog, weaviate-vectorindex} and surveys~\cite{wang2021comprehensive, pan2024survey, pan2024vector}, which restate it as a settled property of these indexes. It also remains largely untested: benchmarks commonly measure the recall--cost trade-off at a single dataset size, not the scaling of cost across sizes at fixed recall~\cite{aumuller2020ann, simhadri2024results, jaasaari2025vibe, wang2021comprehensive}.

\Paragraph{Observations}
We put the folklore to the test and make two observations: (1)~When the dataset size~$N$ is small relative to the data's intrinsic dimensionality, the average search cost of a graph-based index at fixed recall grows in step with $N^{c}$ for a constant $0 < c < 1$. We call this scaling a \emph{Sublinear Power Law}. (2)~Once $N$ is large enough, the growth slows to subpolynomial, consistent with the prevailing poly-logarithmic claim. We verify these observations through extensive experiments: the Sublinear Power Law appears on all eight of our datasets, at recall targets between 90\% and 99\%, under many index configurations, and for hard and easy queries alike (Section~\ref{sec:query_cost}). On two datasets, the data grows past the transition point and search cost begins to slow below the power law~(Sections~\ref{sec:query_cost} and~\ref{sec:billion_scale}).

\Paragraph{Core Contribution: Unifying Theory} By analyzing the anatomy of query cost, we find the mechanism behind these scaling behaviors: search cost is governed by the number of vectors in a small neighborhood around the query. The size of that neighborhood grows exponentially with the intrinsic dimensionality the dataset exhibits at its current size. As long as the dataset sparsely samples its underlying distribution, that dimensionality grows in step with~$\log N$. An exponential in $\log N$ is a power of $N$, so the neighborhood, and hence the cost, grows as a power law. Once the dataset samples the distribution densely, intrinsic dimensionality growth slows and cost growth approaches poly-logarithmic (Section~\ref{sec:anatomy_of_query_cost}). Our unifying theory formalizes this claim. For both exact and bounded-degree constructions\footnote{The \emph{exact construction} applies the index's pruning rule with every other vector as a candidate and keeps every surviving edge. The \emph{bounded-degree construction} additionally caps each vector's degree, as practical indexes do.} of HNSW and Vamana, we prove the Sublinear Power Law and the eventual transition to poly-logarithmic scaling, derive the scale at which that transition occurs, and extend the results to other graph-based indexes (Section~\ref{sec:theory-power-law}).

\Paragraph{Additional Contribution: Cost Models} Under the Sublinear Power Law, we model the query cost as a function of the recall target and the construction parameters, across datasets, for both HNSW and Vamana (Section~\ref{sec:query-cost-model}). For insertions, we model the cost analogously (Section~\ref{sec:insert_cost}). Combined, the models help users navigate the trade-offs between query cost, insertion cost, and recall.

\section{Background}
\label{sec:background}

This section provides background on vector databases, graph indexes, and how they are thought to scale with dataset size. A vector database stores complex objects as points in a high-dimensional space, produced by an embedding model that maps similar objects to nearby points~\cite{wang2024mustapp,jing2025largeapp,gao2023retrievalapp,joshiintroductionapp}. A query is embedded into that same space by the same embedding model, and the search process returns the point(s) closest to it. Formally, given a set of points \(S \subset \mathbb{R}^d\) and a query point \(q\), the goal is to identify the point in \(S\) closest to \(q\). The \(k\)-nearest neighbor variant generalizes this, returning the \(k\) points in \(S\) with the smallest distance to \(q\). 
The embedding model dictates the dimensionality \(d\), which ranges from hundreds to thousands (e.g., 1536 for OpenAI's text-embedding-3-small and 3072 for text-embedding-3-large~\cite{openai2024embeddings}).  The dimensionality must be high enough to represent the many attributes needed to distinguish objects across a large and diverse corpus.

\Paragraph{Curse of Dimensionality \& Exact Search}
In such high-di\-men\-sion\-al spaces, search is hard due to the \emph{curse of dimensionality}~\cite{indyk1998approximateLSH, beyer1999nearest}, a phenomenon whereby most points lie at similar distances from the query, and the nearest neighbor is only slightly closer than the rest~\cite{song2025trim, he2012difficultynearestneighborsearch}.
Because distances concentrate, points do not separate cleanly, so traditional multidimensional indexes (e.g., R-trees~\cite{Rtree} or KD-trees~\cite{bentley1975multidimensionalkdtree}) lose their pruning power~\cite{weber1998quantitativequant, bohm2000cost, li2017fast}.
To distinguish the nearest neighbors from the many nearly equidistant points, the search must examine a large fraction of the dataset. This is expensive~\cite{navarro2002searching} and degrades to a linear scan in the worst~case~\cite{bringmann2021fine}.

\Paragraph{Manifolds \& Intrinsic Dimensionality} Real datasets, despite their high embedding dimension, do not span the entire embedding space. Instead, the data resides near lower-dimensional structures embedded within this space~\cite{fefferman2016testing}. These structures are called \emph{manifolds}~\cite{brown2022verifying}, each of which behaves locally like a flat, Euclidean space. Each manifold's dimensionality is considered the \emph{intrinsic dimensionality} of the data residing near the manifold~\cite{levina2004maximum}. 
What makes search expensive is this intrinsic dimensionality, not the number of dimensions in the embedding. It is smaller than the embedding dimension, but still big enough that exact search stays slow~\cite{pestov2008axiomatic}.

\Paragraph{Approximate Nearest Neighbor Search}
Approximate nearest neighbor (ANN) search sidesteps this expensive query cost by giving up the guarantee of returning the exact nearest neighbors.
This allows the search procedure to apply a rough pruning rule that discards a sizable portion of the search space while ensuring that a nearest neighbor is retained with high probability.
Query accuracy for ANN is commonly formalized in two ways. The first is the classical $(1+\epsilon)$-approximate distance guarantee: the returned point's distance to query $q$ is within a factor of $(1+\epsilon)$ of the nearest neighbor distance. The same relaxation extends to the $k$-nearest neighbor case. This guarantee constrains only the distance, not whether the point is a true nearest neighbor.
The second, \emph{recall}, avoids this issue, and is defined as the fraction of the true nearest neighbors that the search returns. Formally, with $R$ the set of true $k$ nearest neighbors and $R'$ the returned set, recall is $|R \cap R'| / k$. A higher recall value indicates a closer match to the exact search result.
This tolerance for inexact answers motivates a class of data structures, called \emph{vector indexes}, built to answer ANN queries~\cite{wang2021comprehensive, chronis2025filtered}.

\begin{figure}[t]
    \centering
    \includegraphics[width=0.8\columnwidth]{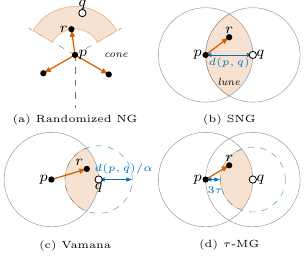}
    \Description{Schematic subfigures, one per graph construction, each showing nodes p, q, and r with a red directed edge from p to r. A shaded region around r, whose shape differs across subfigures, marks the area in which r occludes q and suppresses the candidate edge from p to q.}
    \caption{Edge selection across graph constructions. In each subfigure, $p$ decides whether to add an edge to $q$. The node $r$ already connected to $p$ (red directed edge) lies in the shaded region, where it occludes $q$ and suppresses the edge $p$--$q$.}
    \label{fig:graph_constructions}
\end{figure}

\Paragraph{Graph-based Indexes} Graph-based indexes are widely regarded as the state of the art among vector indexes~\cite{HNSW, DiskANN}. They answer ANN queries by traversing a graph whose nodes are the data points. Edges connect each point to a small set of other points, arranged so that search toward a query approaches its nearest neighbors. All graph-based indexes draw these edges following two core intuitions: (1) each point's edges should lead to points near it, so that once the search closes in on the query, it has easy access to the true nearest neighbors, and (2) the edges should spread across distinct directions, so that from any position some edge makes progress toward the query. 
The two intuitions act in tandem: every point has many nearby points it could connect to, but diversity decides which few of them become edges, sparsifying the graph while preserving navigability. A query that comes from the same distribution as the data lands amid points whose edges already cover its general direction, and is reachable for the same reason the data points are. 
Next, we present the popular graph constructions, beginning with the theoretically grounded designs and moving to the practical relaxations they inspired.

\Paragraph{Randomized Neighborhood Graph}
The \emph{Randomized Neighborhood Graph} (Randomized NG)~\cite{SNG} is one of the earliest graph-based indexes with a theoretical guarantee for ANN search (Figure~\ref{fig:graph_constructions}-a). 
It builds a directed graph over the data points. Specifically, it guarantees that for any query $q$, the search finds a point whose distance to $q$ is within a factor of $1+\epsilon$ of the distance to $q$'s nearest neighbor.
To build this graph, the construction process partitions the space around each point \(p\) into cones with apex \(p\). Within each cone, it scans the enclosed points in random order and links \(p\) to every point closer to \(p\) than all previously scanned points, i.e., every new running minimum. This scan results in an average of $O(\log N)$ neighbors in every cone, at approximately geometrically increasing distances from \(p\). Drawing the edges of a single point requires comparing it against all other points to track the running minimum per cone, so drawing the edges of all \(N\) points costs \(O(N^2)\).

The query cost of a graph-based index factors into the number of search steps and the out-degree at each step. For Randomized NG, both are poly-logarithmic in~$N$. The search takes $O(\log^2 N)$ steps: an ideal straight path to $q$ takes $O(\log N)$ steps, each halving the remaining candidates as in a skip list~\cite{skiplist}, and each ideal step expands into $O(\log N)$ actual steps as the path bends in high dimensions. The out-degree is $O((1/\epsilon)^{d-1}\log N)$: the $(1/\epsilon)^{d-1}$ factor accounts for the number of cones around each point, and each cone contributes $O(\log N)$ neighbors~\cite{SNG}.
Although both factors scale poly-logarithmically with~$N$, the $(1/\epsilon)^{d-1}$ factor grows large in high dimension: smaller $\epsilon$ needs narrower cones, and higher dimension needs more cones to cover all directions. The graph therefore grows dense on both counts, which slows search~\cite{SNG}.

\begin{table}[t]
    \centering
    \scriptsize
    \renewcommand{\arraystretch}{0.95}
    \setlength{\tabcolsep}{3pt}
    \caption{Notations used in the paper.}
    \label{tab:notation}
    \begin{tabular}{ll}
        \toprule
        \textbf{Symbol} & \textbf{Description} \\
        \midrule
        \multicolumn{2}{l}{\textit{Dataset and query}} \\
        $N$                        & Number of vectors in the index \\
        $d_{\mathrm{amb}}$         & Ambient dimensionality of the embedding space \\
        $d_{\mathrm{int}}$         & Intrinsic dimensionality of the dataset \\
        $k$                        & Number of nearest neighbors returned per query \\
        $q$                        & A query vector \\
        \addlinespace
        \multicolumn{2}{l}{\textit{Index parameters}} \\
        $M$                        & Maximum out-degree of each node in the graph \\
        $ef_{\text{construction}}$, $ef_{\text{search}}$ & Candidate list size at build and search time \\
        \addlinespace
        \multicolumn{2}{l}{\textit{Cost metrics}} \\
        $\mathrm{dc}$              & Distance computations per query \\
        $\mathrm{dc}_{\text{build}}$ & Total distance computations to build an index of size~$N$ \\
        \addlinespace
        \multicolumn{2}{l}{\textit{Query-cost model (Eq.~\ref{eq:power-law}, \ref{eq:c_model}, \ref{eq:beta})}} \\
        $c$                        & Query-cost scaling exponent, $\mathrm{dc}\propto N^c$ \\
        $c_0$                      & Dataset-specific baseline exponent \\
        $\alpha$                   & Sensitivity of $c$ to recall demand $\ln(1/\delta)$ \\
        $\gamma_1$, $\gamma_2$     & Reduction in $c$ per $\ln(ef_{\text{construction}})$ and per $\ln(M)$ \\
        $\beta$                    & Search-side exponent, $\mathrm{dc}/q \propto ef_{\text{search}}^{\,\beta}$ \\
        \addlinespace
        \multicolumn{2}{l}{\textit{Build-cost model (Eq.~\ref{eq:build_model})}} \\
        $\kappa$                   & Dataset-specific intercept for build cost \\
        $\beta'$                   & Exponent of $ef_{\text{construction}}$ in build cost \\
        $c'_0$                     & Baseline size-dependence of build cost \\
        $\gamma'_2$                & Interaction coefficient between $\ln M$ and $\ln N$ \\
        \bottomrule
    \end{tabular}
\end{table}

\Paragraph{Sparse Neighborhood Graph} The same paper~\cite{SNG} that introduced Randomized NG also proposed a simpler index, the \emph{Sparse Neighborhood Graph} (SNG), which underlies many practical graph-based indexes today. SNG connects two points $p$ and $q$ only when no third point $r$ to which $p$ is already connected lies in the intersection of the two balls of radius $d(p,q)$ centered at $p$ and $q$. This intersection is called the \emph{lune} of $p$ and $q$ (Figure~\ref{fig:graph_constructions}-b). Intuitively, such a point $r$ provides a closer intermediate step from $p$ toward $q$, so the direct edge from $p$ to $q$ can be removed while preserving navigability. SNG is built by considering, for each point $p$, all other points in increasing order of distance from $p$ and pruning candidates that violate this rule. 
A naive construction takes $O(N^3)$ time, while the paper achieves $O(N^2)$~\cite{SNG} by drawing on algorithms for the \emph{relative neighborhood graph}\footnote{First introduced to capture the local structure of a set of points in a sparse manner.} from the 1980s~\cite{jaromczyk1987note,RNG-org}.

For SNG,\footnote{Proposed as MRNG~\cite{NSG}, with a construction identical~\cite{NSSGfu2021high} to that of SNG~\cite{SNG}.} a later work~\cite{NSG} proves that the search takes in expectation $O(N^{1/d}\log N / \triangle r)$ steps for uniformly distributed data and queries that are dataset points, where $\triangle r$ is the minimum gap between pairwise distances in the dataset.
A later paper~\cite{tauMG} proves that, under the same assumptions, $\triangle r$ shrinks as $O(N^{-1/d})$, bounding the steps by $O(N^{2/d}\log N)$. The out-degree, in turn, is independent of~$N$ and hence enters the query time only as a constant, but it is exponential in $d$: the SNG paper bounds it by the number of caps of angular diameter \(\pi/3\) fitting on the \(d\)-dimensional sphere, i.e., $2^{O(d)}$, and the empirical out-degree grows accordingly~\cite{SNG}. This makes each step, and hence the search, expensive in high dimensions even when the number of steps remains logarithmic.

\Paragraph{Vamana} Inspired by SNG, Vamana relaxes the pruning rule with a slack parameter \(\alpha \ge 1\)~\cite{DiskANN} (Figure~\ref{fig:graph_constructions}-c).
Vamana connects two points \(p\) and \(q\) only when no point \(r\) already connected to \(p\) comes closer to \(q\) than \(p\) by a factor of at least \(\alpha\), i.e., \(\alpha \cdot d(r,q) < d(p,q)\). At \(\alpha = 1\), this is exactly the pruning rule in SNG.
The intuition behind the slack parameter is that every node has an edge either to the query's nearest neighbor or to a point at least \(\alpha\) times closer to the query than the node itself is, so each greedy step cuts the remaining distance by the constant factor \(\alpha\).
Guaranteeing this progress costs extra edges: a larger \(\alpha\) makes the pruning condition harder to satisfy, so fewer edges are removed, increasing the graph's out-degree.

A later work~\cite{indyk2023worst} proves that greedy search reaches an $\left(\frac{\alpha+1}{\alpha-1}+\epsilon\right)$-approximate nearest neighbor in $O(\log_\alpha \tfrac{\Delta}{(\alpha-1)\epsilon})$ steps, where ~$\Delta$ is the ratio of the largest to smallest pairwise distance in the dataset (i.e., the \emph{spread}). 
This ratio stays small in practice even at large~$N$~\cite{indyk2023worst,song2025trim}.
Each step pays the out-degree of the visited point, bounded by $O((4\alpha)^d \log \Delta)$, again exponential in the dimensionality.\footnote{Here $d$ is the doubling dimension, a measure of the intrinsic dimensionality.} When~$\Delta$ grows at most polynomially with~$N$, as under a uniform distribution, both~$\log \Delta$ factors become $O(\log N)$, and multiplying steps by out-degree gives a query time poly-logarithmic in~$N$. Also, the construction process, as in SNG, prunes each point against the entire dataset, incurring a total cost at least quadratic in~$N$.

\Paragraph{\(\tau\)-Monotonic Graph} Also inspired by SNG, the \(\tau\)-Monotonic Graph (\(\tau\)-MG)~\cite{tauMG} extends SNG's guarantees to query points within distance \(\tau\) of some dataset point, for a fixed \(\tau\). The construction process connects any two points within distance \(3\tau\), and for farther pairs applies an SNG-like rule on the lune shrunk by \(3\tau\) (Figure~\ref{fig:graph_constructions}-d). These extra edges densify SNG and guarantee a path that moves at least \(\tau\) closer to the query at every step. Assuming uniformly distributed data and a query within \(\tau\) of some dataset point, the paper proves the search reaches the nearest neighbor in \(O(N^{1/d}\log^2 N)\) steps. 
The paper does not analyze the out-degree, treating it as a constant. In fact, because all points within \(3\tau\) are connected, the out-degree is \(\Omega((3\tau)^d)\), so the constant hidden in its query cost is exponential in the dimensionality. As in SNG, construction cost is at least quadratic in \(N\).

\begin{table}[t]
    \centering
    \scriptsize
    \setlength{\tabcolsep}{2pt}
    \caption{Datasets used in the experiments.}
    \label{table:datasets}
    \begin{tabular}{lcccccc}
        \toprule
        \textbf{Dataset} & \textbf{Dim.} & \textbf{Int. Dim. @ 1M} & \textbf{Type} & \textbf{Metric} & \textbf{Model} & \textbf{Sampled From} \\
        \midrule
        SIFT             & 128          & 19.2 & Image         & L2              & SIFT descriptors~\cite{lowe2004distinctive}  & SIFT1B~\cite{jegou2011searching} \\
        DEEP             & 96           & 20.7 & Image         & L2              & GoogLeNet~\cite{szegedy2015going}            & DEEP1B~\cite{babenko2016efficient} \\
        SpaceV           & 100          & 29.2 & Text          & L2              & SpaceV Superior~\cite{shan2021glow}          & SpaceV1B~\cite{spacev1b} \\
        Wiki             & 768          & 30.1 & Text          & L2              & MPNet~\cite{song2020mpnet}                    & Wiki-all~\cite{wikiall} \\
        GloVe            & 100          & 32.4 & Text          & Angular         & GloVe~\cite{pennington2014glove}             & GloVe-100~\cite{pennington2014glove} \\
        Rand64           & 64           & 40.2 & Synthetic     & Angular         & ---                                          & $\mathrm{Unif}(S^{63})$\tablefootnote{Generated by sampling each vector i.i.d.\ from $\mathcal{N}(0,I_{64})$ and normalizing to unit length, which yields a uniform distribution on $S^{63}$.} \\
        GIST             & 960          & 42.5 & Image         & L2              & GIST descriptors~\cite{oliva2001modeling}    & GIST1M~\cite{jegou2010product} \\
        OpenAI           & 1536         & 46.5 & Text          & L2              & ada-002~\cite{openai2022ada002}              & OpenAI-ArXiv~\cite{openai-arxiv-dataset} \\
        \bottomrule
    \end{tabular}
\end{table}

\Paragraph{Practical Graph-based Indexes}
The quadratic construction cost of the above graphs is the core obstacle to their practical use at scale~\cite{tauMG,SNG,NSG}. Modern practical indexes lower this cost by restricting the set of points to which each point can connect. Rather than treating every point as a candidate for a point's edges, they apply a greedy \emph{beam search} to find a much smaller set of nearby points, then apply the pruning rule within that smaller set. More precisely, when inserting a new point, the index runs a beam search on the current graph for the new point. Beam search maintains a candidate list of fixed size, here $ef_\text{construction}$, and starts from one or more entry points. It repeatedly expands the closest unexpanded candidate, keeping only the closest points found so far in the list. Once no unexpanded candidate remains, the list is pruned locally and the new point is connected to at most $M$ of the closest survivors. The same beam search procedure is used at query time, with a candidate list of size $ef_\text{search} \geq k$, and the top-$k$ results are selected from it.

The practical indexes differ mainly in the pruning rule each applies to this candidate list rather than to the full dataset.
HNSW~\cite{HNSW} organizes nodes into hierarchical layers, assigning each node to higher layers with exponentially decaying probability, and applies the SNG rule within each layer. The practical version of Vamana~\cite{DiskANN} applies the \(\alpha\)-based pruning rule; henceforth, we refer to this version of Vamana simply as Vamana.
The HNSW paper claims that its query cost scales logarithmically with dataset size while its construction cost stays below quadratic, at $O(N \log N)$~\cite{HNSW}.
The paper that introduced Vamana~\cite{DiskANN} argues that greedy search on its exact construction takes logarithmically many steps but leaves the query cost scaling of its practical version unstated; benchmarks nonetheless place it alongside HNSW as state-of-the-art~\cite{DiskANN, aumuller2020ann}. These logarithmic scalability claims and similar claims from the same line of work~\cite{NSG, NSSGfu2021high, tauMG, malkov2014approximate, alphaCG} have since propagated into vendor documentation and industry blogs, which restate them as a settled property of these indexes~\cite{elastic-labs-knn, milvus-index, opensearch-knn-tuning, pinecone_hnsw, redis-hnsw-blog, weaviate-vectorindex}.
In the next section, we evaluate how these indexes scale and bridge this gap between claimed and measured scaling both theoretically and~empirically.

\begin{figure}[t]
    \centering
    \includegraphics[width=1\columnwidth]{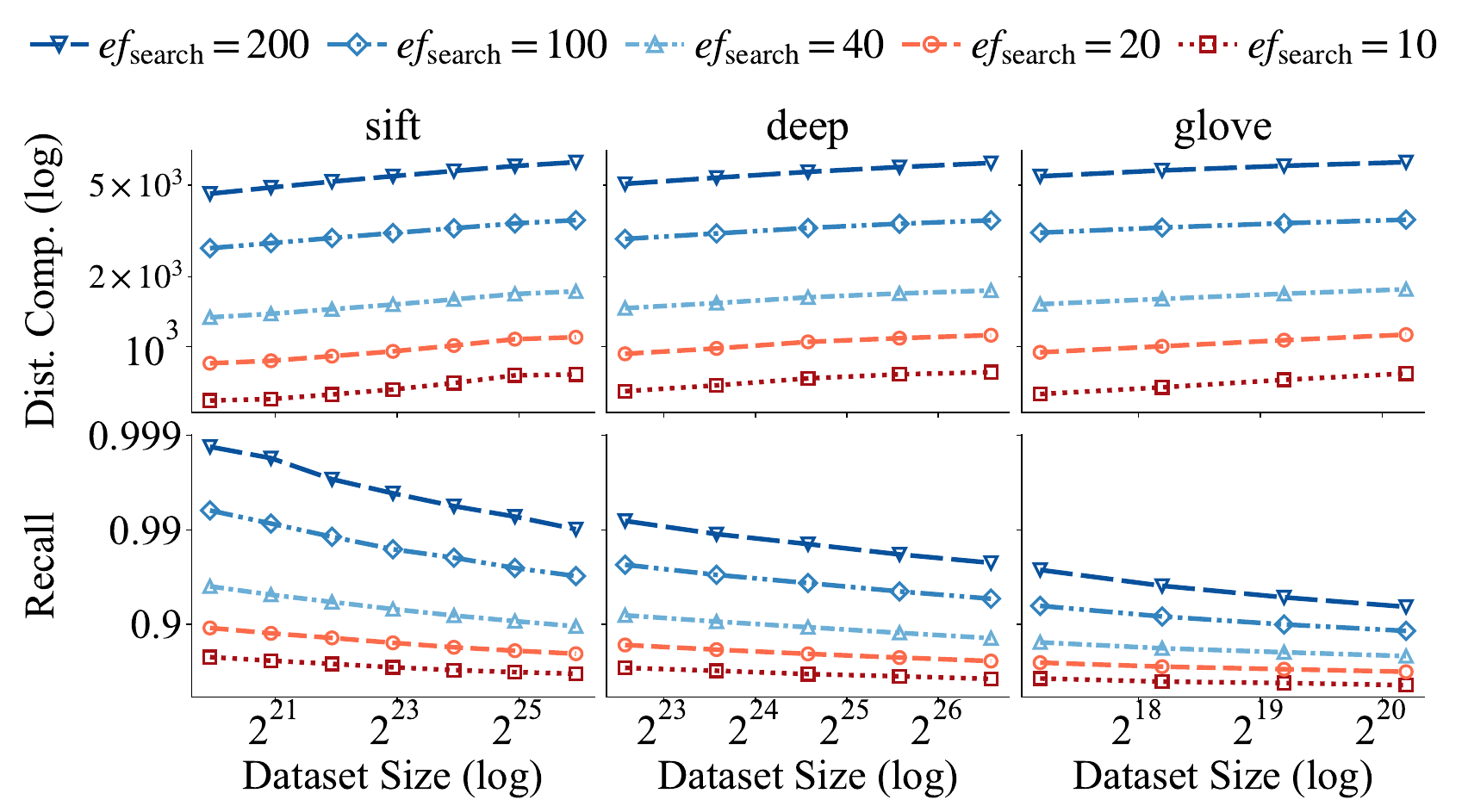}
    \Description{Two stacked line charts over dataset size, with one line per search beam width value. The top chart shows distance computations per query increasing as the dataset grows; the bottom chart shows recall decreasing.}
    \caption{As the data grows, distance computations per query rise (top) and recall drops (bottom). Here, we use the HNSW index with $ef_\text{construction}=200$, $M=32$, and $k=10$, with one curve per $ef_\text{search}$ value.}    
    \label{fig:hnsw_fixed_setup}
\end{figure}

\begin{figure*}[t]
    \centering
    \includegraphics[width=1\textwidth]{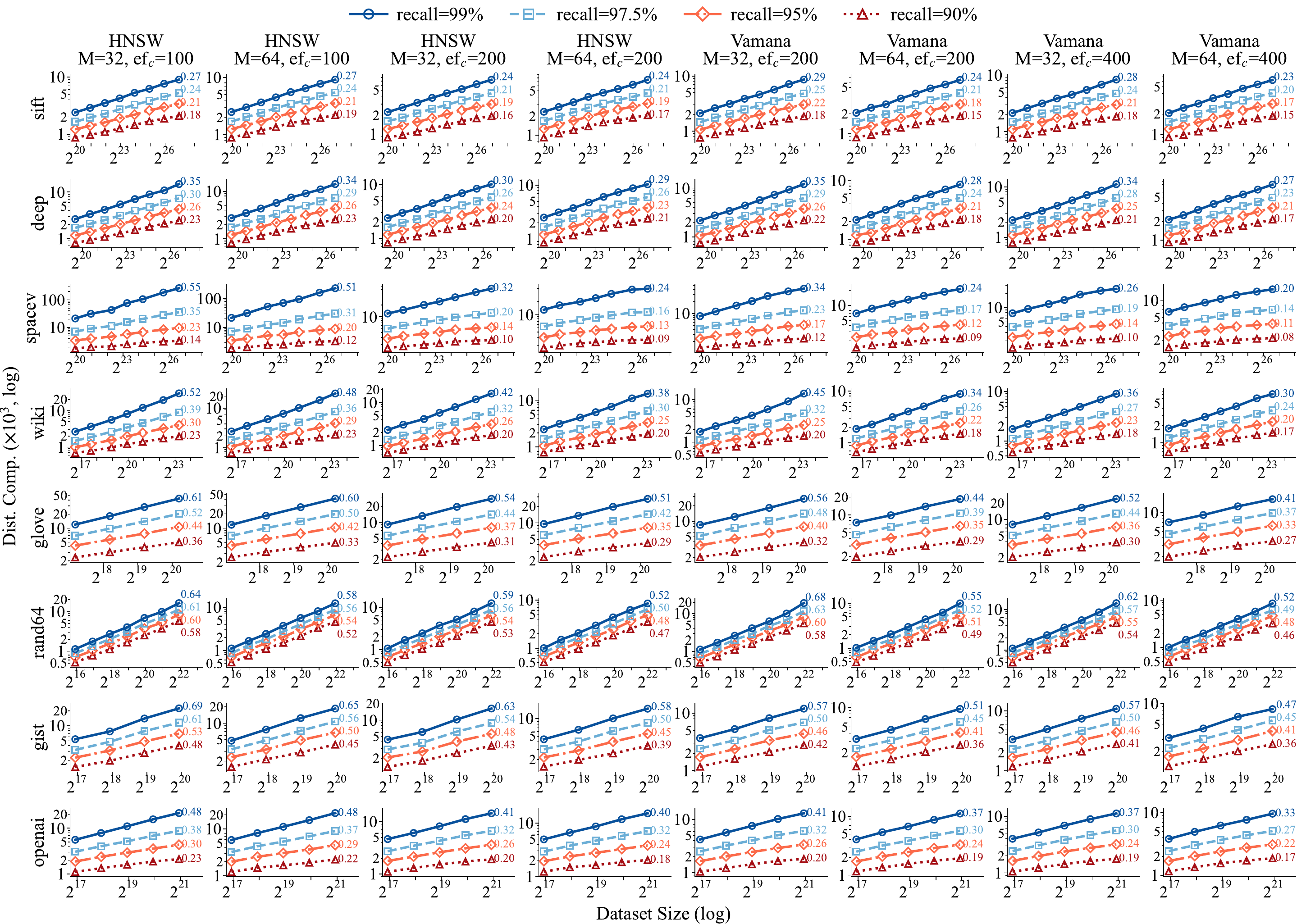}
    \Description{A row of log-log line plots, one per dataset, showing distance computations per query versus dataset size at several fixed recall levels. In every subfigure the curves form nearly straight lines, consistent with a Sublinear Power Law.}
    \caption{When recall is fixed, the number of distance computations (y-axis) increases with dataset size (x-axis) across the datasets from Table~\ref{table:datasets}, forming a linear pattern in the log-log diagram. This is consistent with a sublinear power-law. 
    }
    \label{fig:hnsw_faiss_scale}
\end{figure*}
	
\section{How Graph-Based Indexes Scale}
\label{sec:scaling_graph}

In this section, we study empirically how query cost scales as data grows, focusing on HNSW and Vamana\footnote{We use the in-memory version of Vamana, built with \texttt{build\_memory\_index}~\cite{diskann-lib}.}, the most widely used graph-based indexes~\cite{milvus-index, pan2024survey}. Later in Section~\ref{sec:why-a-power-law} we show that our findings generalize to other graph-based indexes.

\subsection{The Problem: Recall Deterioration}
\label{sec:det_recall}

Our first experiment verifies a known phenomenon: when the index parameters are held fixed, recall degrades as the dataset grows~\cite{duan2025pgtuner, manohar2024parlayann, zhao2026grab}.
To simulate data growth, we draw uniformly random subsets of increasing sizes from each of the datasets in Table~\ref{table:datasets} up to their full sizes.
Each subset is shuffled before index construction. When a dataset provides a canonical query set, we use it as our query workload. Otherwise, we reserve 10,000 uniformly sampled points as queries, excluding them from the set of points over which the index is built. 
To compute the true nearest neighbors for each query (i.e., its ground truth), we perform an exact search by linearly scanning the entire dataset. We then execute the query using the index and compare the returned results against the ground truth to compute recall.
We measure query cost in terms of the number of distance computations (\emph{dc}) per query as this is the dominant cost in graph traversal~\cite{NSG, yang2025effective, gao2023high}.
By so doing, we avoid conflating algorithmic performance with implementation idiosyncrasies (cache evictions, prefetching, hardware-optimized distance computations, etc.). All experiments in the paper follow this setup.

Figure~\ref{fig:hnsw_fixed_setup} shows the results on HNSW using three representative datasets (the remaining ones behave similarly). Each column represents a different dataset (SIFT, DEEP, and GloVe). The top row reports distance computations and the bottom recall. Each subfigure shows one curve per setting of the $ef_\text{search}$ parameter.
For each of these curves, we observe that distance computations per query slightly increase with dataset size~$N$ (top row), while recall drops significantly (bottom row). The reason for this behavior is that as the data grows, more points lie close to the query. Beam search therefore converges more slowly, and the true nearest neighbors face more competition for a place in the candidate list, as suggested in prior work~\cite{manohar2024parlayann}. This recall degradation is problematic for applications such as retrieval-augmented generation (RAG)~\cite{karpukhin2020dense} and recommendation systems~\cite{covington2016deep}, which require stable recall even as datasets~grow.

\subsection{The Impact of Fixing Recall}
\label{sec:query_cost}

The experiment in the previous section raises a natural question: can recall be kept stable as the dataset grows? Although prior work has not addressed this question directly, it has considered the related problem of meeting a specified recall target. The standard approach is to periodically measure recall using exact search~\cite{qdrant_recall, vespa_recall, elastic_recall} and adjust search parameters accordingly. If the measured recall falls below the target, the user increases the $ef_\text{search}$ parameter, measures recall again, and repeats this process until the target is reached~\cite{manohar2024parlayann, turbopuffer2024recall, chatzakis2025darth}.
Recent work in academia as well as in industry (e.g., Zilliz Cloud~\cite{zilliz_autoindex} and Turbopuffer~\cite{turbopuffer2024recall}) improves upon this approach by using statistical or learned models of the search effort required to achieve a recall target~\cite{chatzakis2025darth, zhang2026distribution, earlyterm2020}.
Once recall is held fixed, another core question emerges: \emph{how does query cost scale as the dataset grows?}

\Paragraph{Experimental Setup} To hold recall fixed, we set $ef_{\text{search}}$ to the smallest value that attains a given fixed recall target.
Figure~\ref{fig:hnsw_faiss_scale} reports the average number of distance computations per query across a grid of datasets, recall targets, and build settings. Rows correspond to different datasets. We progressively grow each dataset by increasing its number of points up to its original full size. For the substantially larger SIFT, DEEP, and SpaceV\footnote{This dataset contains many duplicate points. To compare fairly against other datasets, we de-duplicate it and sample the query set exclusively from the base dataset.} datasets, we cap the size at 100M points to keep the total experimental runtime to approximately one week. Recall targets of 90\%, 95\%, 97.5\%, and 99\% appear as separate curves in each subfigure. Each column corresponds to a different index configuration. HNSW occupies the left four columns and Vamana the right four, with the maximum out-degree $M\in\{16, 32, 64, 128\}$ and construction list size $ef_{\text{construction}}\in\{100, 200, 400, 800\}$ varying across the columns.\footnote{For Vamana, $\alpha$ is a build parameter fixed to the default 1.2~\cite{DiskANN}; testing $\alpha\in\{1.5,2\}$ left scalability unchanged, as out-degree typically saturates at cap $M$.} Due to limited space, we show four of the sixteen benchmarked $(M, ef_{\text{construction}})$ configurations per index; the remaining settings behave similarly.

\Paragraph{Results}
On the log-log axes in Figure~\ref{fig:hnsw_faiss_scale}, most curves are approximately linear as the dataset grows, indicating that query cost follows a \textbf{\emph{Sublinear Power Law}}:
\begin{equation}
    \label{eq:power-law}
    \text{dc} \propto N^c.
\end{equation}
Here, $0 < c < 1$, and the value at the right end of each curve reports the fitted exponent~$c$.

\begin{figure}[t]
    \centering
    \includegraphics[width=0.99\columnwidth]{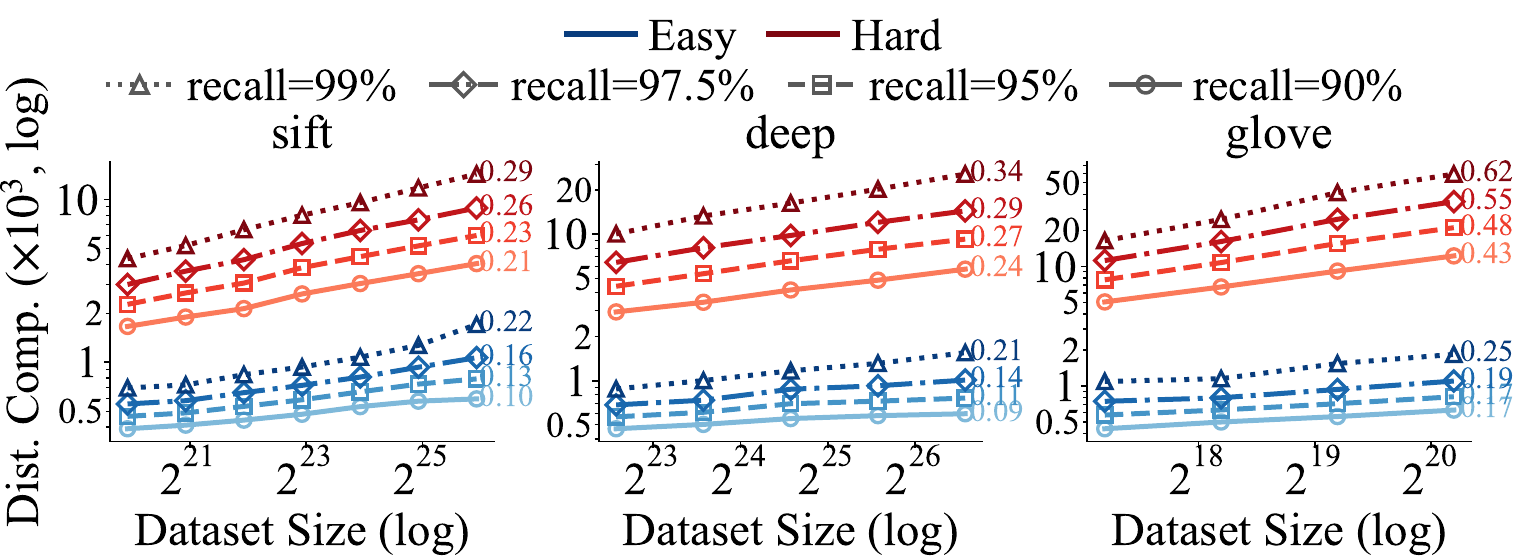}
    \Description{Log-log line plot of distance computations per query versus dataset size, with separate lines for hard and easy queries. Both sets of lines are straight and nearly parallel, indicating power-law scaling for both query groups.}
    \caption{Both hard and easy queries follow a power law scaling. We use the same HNSW setup as Figure~\ref{fig:hnsw_fixed_setup}.}
    \label{fig:hnsw_hardness}
\end{figure}

Figure~\ref{fig:hnsw_faiss_scale} also indicates the factors that impact the exponent~$c$ in Equation~\ref{eq:power-law}. The exponent increases with higher recall targets, since they leave less margin for missing true nearest neighbors and force the search process to traverse more of the graph. 
Higher values of $ef_{\text{construction}}$ and $M$ reduce $c$ by yielding graphs in which the beam search reaches the query's nearest neighbors in fewer hops.
We also measured $k \in \{1, 5, 10, 20\}$ and found that varying it does not materially change $c$, so we report only $k=10$. The nearest neighbors lie at nearly the same distance from the query, so retrieving all $k$ is about as hard as retrieving one. We give more intuition for this in Section~\ref{sec:anatomy_of_query_cost}. 

Beyond these parameters, the exponent tends to increase with dataset hardness. The rows in Figure~\ref{fig:hnsw_faiss_scale} are ordered by increasing intrinsic dimensionality, a proxy for hardness, with estimates listed in Table~\ref{table:datasets}. For example, GloVe has higher intrinsic dimensionality than SIFT (32.4 vs.\ 19.2) and a higher exponent at the same recall level. This suggests that beam-search cost scales more steeply with~$N$ on harder datasets.

\Paragraph{Varying Query Hardness} As discussed in Section~\ref{sec:background}, queries within a workload can vary substantially in hardness~\cite{wang2024steiner,chatzakis2025darth,earlyterm2020}, with queries in regions of higher intrinsic dimensionality generally being more expensive to answer~\cite{aumuller2021lid}. To test whether the Sublinear Power Law holds across query hardness levels within the same workload, we rank queries by intrinsic dimensionality\footnote{Later, in Section~\ref{sec:anatomy_of_query_cost}, we show how to measure intrinsic dimensionality using LID; Figure~\ref{fig:lid_main} reports the resulting estimates.} and rerun the experiment separately on the easiest and hardest 20\%. On SIFT, DEEP, and GloVe, the hard group's intrinsic dimensionality exceeds that of the easy group by at least $7.9$, $9.6$, and $11.0$ points, respectively. Figure~\ref{fig:hnsw_hardness} shows that while the two groups have different exponents, both still follow the Sublinear Power~Law.

\Paragraph{Summary} Our experimental results so far challenge the conventional logarithmic characterization of graph-based indexes from Section~\ref{sec:background}. Under logarithmic scaling, each doubling of~$N$ would add only a constant number of distance computations, causing the curves in Figures~\ref{fig:hnsw_faiss_scale} and~\ref{fig:hnsw_hardness} to bend downward rather than remain straight. Instead, the Sublinear Power Law describes most of the datasets, index configurations, and query hardness levels we measure.

That said, we also observe exceptions to the Sublinear Power Law. For SpaceV and OpenAI, some of the curves in Figure~\ref{fig:hnsw_faiss_scale} seem to drift very slightly below their fitted power laws, as if the scaling were slowing down. Only larger dataset sizes can tell whether the drift is real. OpenAI already uses its full base dataset, whereas SpaceV can grow by another order of magnitude; we explore this~next.

\subsection{Scaling with Larger Data}\label{sec:billion_scale}

\begin{figure}[t]
    \centering
    \includegraphics[width=1\columnwidth]{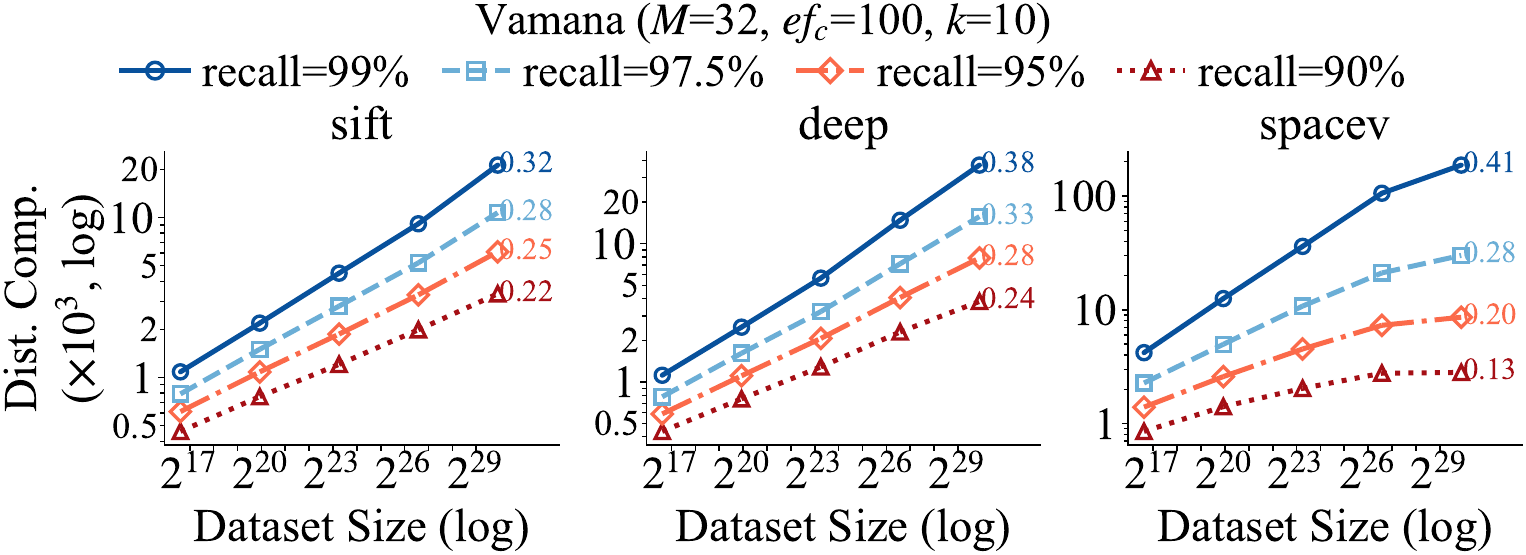}
    \Description{Log-log line plot of distance computations versus dataset size up to one billion points for the SIFT, DEEP, and SpaceV datasets. The SIFT and DEEP curves stay straight throughout, while the SpaceV curve flattens as the dataset grows.}
    \caption{The power law persists to billion scale on SIFT and DEEP, while SpaceV's slope decreases as the dataset grows.}
    \label{fig:billion-scale-test}
\end{figure}

The SIFT, DEEP, and SpaceV datasets each have a base collection size of one billion points, but at that scale the index for a single configuration occupies 300\,GB--1\,TB and takes several days to build~\cite{singh2021freshdiskann}. We therefore choose one representative configuration ($ef_{\text{construction}}=100$, $M=32$) and repeat the experiment with Vamana\footnote{We use Vamana here because it has a smaller memory footprint than HNSW and therefore lets us scale to the largest available dataset sizes.}. We grow each of the three datasets from 100K points to their full base size of one billion points.

Figure~\ref{fig:billion-scale-test} shows that on SIFT and DEEP, the same Sublinear Power Law continues to hold as the data grows. The query cost of SpaceV, however, behaves differently: it slopes downwards as the dataset grows. 
Thus, while the scalability of graph-based indexes exhibits clear empirical patterns, two questions remain: what explains the initial sublinear power-law scaling across all workloads, and why does this scaling eventually slow for some workloads?

\section{Unifying Theory}
\label{sec:why-a-power-law}

We now develop a unifying theory that explains the observations so far: when recall is held fixed, query cost initially scales as a Sublinear Power Law and then slows as the dataset grows. The theory suggests that the slowdown observed for SpaceV and OpenAI in the previous section is not specific to these datasets; we would observe the same transition on the other datasets if they could be scaled to sufficiently large sizes.

\subsection{Anatomy of Query Cost}
\label{sec:anatomy_of_query_cost}
\Paragraph{Search Phases}
Conceptually, beam search on graph-based indexes such as HNSW or Vamana involves two phases. The first, which we call the \emph{shortcut phase}, navigates from the entry point toward the query's neighborhood: at each step, the search picks the candidate closest to the query and checks its neighbors, moving steadily inward~\cite{HNSW, wang2021comprehensive}. In this phase, beam search behaves essentially like greedy search, since the best candidate improves at every step. The second, which we call the \emph{exploration phase}, searches the neighborhood of the query for its true $k$-nearest neighbors: the search now moves outward, expanding candidates it had passed over and branching into paths it initially skipped, until it can no longer make useful headway. We now study the costs of these two phases.

\Paragraph{Exploration Phase} Across all benchmarks in Figures~\ref{fig:hnsw_faiss_scale}--\ref{fig:billion-scale-test}, we measure the proportion of search cost spent in the shortcut phase versus the exploration phase. In every case, the exploration phase accounts for at least 80\% of the cost, consistent with prior observations~\cite{chatzakis2025darth, lu2026terminus, chen2023finger, lin2019graph}. This phase is expensive because reaching the query vicinity does not imply that all true nearest neighbors are easily reachable. In high dimensions, distances concentrate~\cite{song2025trim}, so points close to the same query need not be close to one another. At the same time, bounded node degree limits local connectivity to control memory and per-hop search cost. Together, these effects prevent the query neighborhood from forming a tightly connected region~\cite{diwan2024navigable,cai2013distributions}, forcing nearby points to be reached through roundabout rather than direct paths and requiring exploration beyond the true nearest neighbors to achieve a recall target.

\Paragraph{Modeling Exploration Cost} We now construct a mathematical model of the cost of the exploration phase, aiming to understand why the Sublinear Power Law arises. The model represents the query's neighborhood as two nested hyper-balls (multi-dimensional balls) centered at the query point $q$. The \emph{inner hyper-ball} extends from the query to its $k$-th nearest neighbor and thus has radius $r_k$; the \emph{outer hyper-ball} extends to distance $(1+\varepsilon)r_k$. Here, $1+\varepsilon$ is the smallest factor for which extensively searching the outer hyper-ball visits the true $k$-nearest neighbors with probability equal to the recall target. As the dataset grows, the distances from the query to the points in its neighborhood shrink by a common factor. The inner and outer hyper-balls therefore contract by the same factor, leaving their radii ratio (i.e., $1+\varepsilon$) unchanged (we return to this point later in Discussion~\ref{disc:eps}). 
Henceforth, we define the exploration phase precisely as the distance computations the search performs within the outer hyper-ball, and the shortcut phase as the rest.
The cost of the exploration phase is governed by the number of points inside the outer hyper-ball. Hence, the key question is: \emph{how does the number of points in the outer hyper-ball grow with the dataset size~$N$?}

\Paragraph{Local Uniformity on a Manifold} To answer the question above, we adopt a standard assumption from the literature: real embedding datasets exhibit complex structure at the global scale, yet if we zoom closely enough into any subspace, the distribution within it is approximately uniform~\cite{prokhorenkova2020graph,aumuller2021lid}. More precisely, real-world datasets tend to concentrate across lower-dimensional manifolds. Within a sufficiently small neighborhood around the query, the data behaves as if its density were approximately uniform~\cite{mcinnes2018umap,houle2017local,ma2025graph}. The intuition is that real-world data items vary along continuous attributes, such as the visual style of an image or the topical blend of a document. Embedding models are trained to map similar data items to nearby vectors, so the variation carries over to the vector space, where density changes gradually rather than abruptly. Within a small enough neighborhood, points can therefore be treated as uniform on their manifold. For example, in Figure~\ref{fig:sphere-caps}, the dataset forms a complex manifold in three dimensions, yet the subspace we zoom into is two-dimensional with approximately uniform points. We thus model the points in both the inner hyper-ball, which contains the nearest neighbors, and the outer hyper-ball as uniformly distributed around the query; we call this the \emph{local uniformity model}. 
Figure~\ref{fig:visited_node_distr} later confirms that the outer hyper-ball is indeed local: the search rarely goes far beyond the nearest neighbors.

\Paragraph{Assumptions} Beyond the local uniformity model, our analysis in this section adopts three assumptions. (1)~The query is drawn from the same distribution as the data. (2)~The dataset is dense enough that the hyper-balls' projections onto the manifold are themselves approximately lower-dimensional hyper-balls. (3)~The distance metric is Euclidean. Discussion~\ref{disc:assumptions} revisits what happens when each of these assumptions is dropped.

\begin{figure}[t]
    \centering
    \includegraphics[width=1\columnwidth]{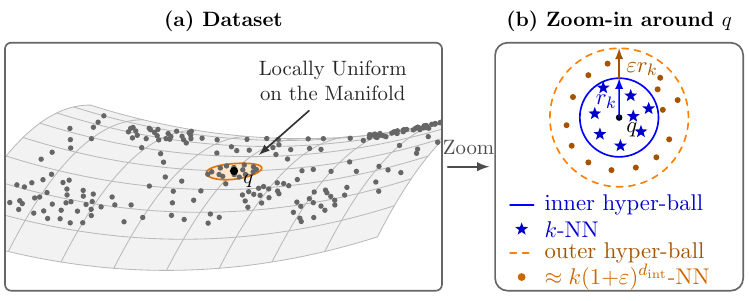}
    \Description{Two schematic subfigures. Subfigure (a) shows a curved manifold in the ambient space with a query point q, where the small neighborhood around q appears uniformly sampled even though the whole dataset is not. Subfigure (b) shows two concentric balls around q: an inner ball of radius r_k containing the k nearest neighbors, and a slightly larger outer ball obtained by expanding that radius by a factor of one plus epsilon.}
    \caption{(a) The data is not globally uniform, but the query $q$'s neighborhood is roughly uniform on the manifold. (b)~The $k$-NNs of $q$ lie within radius $r_k$; expanding this radius by a factor $(1+\varepsilon)$ captures about $k(1+\varepsilon)^{d_{\mathrm{int}}}$ points.}
    \label{fig:sphere-caps}
\end{figure}

\Paragraph{Dimensionality and Ball Volume}
Under the above assumptions, the cost of the exploration phase is proportional to the number of points inside the outer hyper-ball (we revisit this point in Section~\ref{sec:theory-power-law}).\footnote{Similar measures have been used to estimate query hardness~\cite{zoumpatianos2018generating, ceccarello2025evaluating}.} The number of such points is proportional to the outer hyper-ball's volume on the manifold. The volume, in turn, is governed by the manifold's dimensionality, i.e., the intrinsic dimensionality near the query point ($d_{\mathrm{int}}$), rather than by the ambient dimensionality of the embedding space ($d_{\mathrm{amb}}$). The following theorem counts the points inside the outer hyper-ball in terms of its volume on the manifold.

\begin{theorem}[\textsc{Scaling with Intrinsic Dimensionality}]
    \label{thm:local-hyperball-scaling}
Let $B(q,r)$ be the hyper-ball of radius $r$ around query $q$ on the manifold, and let $C(r)$ be the number of data points inside $B(q,r)$. For any radius $r$ and any $\varepsilon$ such that the local uniformity model holds within $B(q,(1+\varepsilon)r)$,
\begin{equation*}
    C((1+\varepsilon)r) \approx (1+\varepsilon)^{d_{\mathrm{int}}}\, C(r).
\end{equation*}
    In particular, if $r_k$ is the query's distance to its $k$-th nearest neighbor, then the expected number of points inside the outer hyper-ball is
    \begin{equation*}
        C((1+\varepsilon)r_k)
        \approx
        k(1+\varepsilon)^{d_{\mathrm{int}}}.
    \end{equation*}
\end{theorem}

\begin{proof}
    A $d_{\mathrm{int}}$-dimensional hyper-ball has volume proportional to $r^{d_{\mathrm{int}}}$~\cite{ball1997elementary}.
    By assumption, the local uniformity model holds within $B(q,(1+\varepsilon)r)$, and $B(q,r)$ is contained in it, so the expected number of points in each of the two hyper-balls scales with its volume on the manifold (Figure~\ref{fig:sphere-caps}-b). Specifically, we have
    \begin{equation*}
        \frac{C((1+\varepsilon)r)}{C(r)}
        \approx
        \frac{((1+\varepsilon)r)^{d_{\mathrm{int}}}}{r^{d_{\mathrm{int}}}}
        =
        (1+\varepsilon)^{d_{\mathrm{int}}}.
    \end{equation*}
    Setting $r=r_k$ gives $C(r_k)\approx k$, so
    $C((1+\varepsilon)r_k)
    \approx
    k(1+\varepsilon)^{d_{\mathrm{int}}}.$
\end{proof}

The base $1+\varepsilon$ can be rewritten as $N^{\log(1+\varepsilon)/\log N}$. Substituting this into the factor $(1+\varepsilon)^{d_{\mathrm{int}}}$ and simplifying the exponent yields the following:
\begin{corollary}[\textsc{Scaling with Dataset Size}]
    \label{cor:local-hyperball-scaling-n}
    Under the assumptions of Theorem~\ref{thm:local-hyperball-scaling}, the expected number of data points inside the outer hyper-ball around the query can be written as
    \begin{equation*}
        C((1+\varepsilon)r_k)
        \approx
        kN^{\frac{d_{\mathrm{int}}}{\log N}\log(1+\varepsilon)}.
    \end{equation*}
\end{corollary}
The exponent of~$N$ in Corollary~\ref{cor:local-hyperball-scaling-n} depends on $d_{\mathrm{int}}$ through the ratio $d_{\mathrm{int}}/\log N$, so applying it requires understanding how $d_{\mathrm{int}}$ behaves as~$N$ grows. 

\newcommand{\axisrule}{\leaders\hrule height 2.15pt depth -1.35pt\hfill}
\newcommand{\axisfill}{\textcolor{black!70}{\axisrule}}
\newcommand{\axisend}{\textcolor{black!70}{\axisrule\kern-3pt$\boldsymbol{\succ}$}}
\begin{table}[t]
\centering
\scriptsize
\caption{The order of the query cost falls as $N$ grows.}
\label{tab:regime-costs-uniform}
\renewcommand{\arraystretch}{1.2}
\setlength{\tabcolsep}{2pt}
\begin{threeparttable}
\begin{tabular}{@{}lcccccccc@{}}
\toprule
& \multicolumn{3}{c}{\textsc{Sparse}: $d_{\mathrm{int}} = \Theta(\log N)$}
&
& \multicolumn{3}{c@{}}{\textsc{Dense}: $d_{\mathrm{int}} = o(\log N)$} \\
\cmidrule(lr){2-4} \cmidrule(l){6-8}
$N\colon$
& \axisfill & $O(M)$\tnote{$\ast$}
& \axisfill & $2^{\Theta(d_{\mathrm{int}})}$
& \axisfill & $2^{2^{\Theta(d_{\mathrm{int}})}}$
& \axisend \\
\midrule
Shortcut     & $O(1)$      & & $N^{\Theta(1)}$ & & $N^{o(1)}$ & & $\Theta(\text{poly}\log N)$\tnote{$\dagger$} \\
Exploration  & $\Theta(N)$ & & $N^{\Theta(1)}$ & & $N^{o(1)}$ & & $O(\text{poly}\log N)$ \\
\midrule
Total        & $\Theta(N)$ & & $N^{\Theta(1)}$ & & $N^{o(1)}$ & & $\Theta(\text{poly}\log N)$ \\
\bottomrule
\end{tabular}
\begin{tablenotes}[para,flushleft]
\item[$\ast$] For $M = 2^{o(d_{\mathrm{int}})}$.
\item[$\dagger$] For some graph-based indexes, these expressions are lower bounds on the cost that may not be tight. For example, as argued in Section~\ref{sec:theory-dense}, SNG's cost is at least~$N^{(1-o(1))/d_{\mathrm{int}}}$, which is larger than the poly-logarithmic term in the right column. Nevertheless, the table's expressions are tight for the exactly constructed version of Vamana, intuitively because it behaves like a skip list.
\end{tablenotes}
\end{threeparttable}
\end{table}

\Paragraph{LID} We use the local intrinsic dimensionality (LID) as our empirical proxy for~$d_{\mathrm{int}}$, as it is computationally feasible and standard in practice~\cite{houle2017local,amsaleg2015lid,aumuller2021lid}. 
LID measures intrinsic dimensionality by how tightly the query's distances to its nearest neighbors concentrate. The more these distances concentrate around a common value, the higher the intrinsic dimensionality. We estimate LID using the maximum-likelihood estimator~\cite{levina2004maximum, aumuller2021lid}
\[
\widehat{\operatorname{LID}}_k(q)
=
-\left(
\frac{1}{k}\sum_{i=1}^{k}\log \frac{r_i}{r_k}
\right)^{-1}.
\]
Here, $r_i$ is the query's distance to its $i$-th nearest neighbor. 

Figure~\ref{fig:lid_main} reports how LID changes with~$N$ across datasets and values of~$k$. For most datasets, LID follows an approximately straight line when plotted against a logarithmic dataset size, meaning LID grows in step with $\log N$. 
Meanwhile, the LID of OpenAI still increases, but with a gradually decreasing slope. The decreasing slope matches Figure~\ref{fig:hnsw_faiss_scale}, where some of OpenAI's curves deviate slightly from the Sublinear Power Law. SpaceV departs from the linear trend even more significantly than OpenAI: its LID rises at first but then almost flattens. This behavior of SpaceV is consistent with Figure~\ref{fig:billion-scale-test}, where it deviates from the Sublinear Power Law.

\begin{figure}[t]
    \centering
    \includegraphics[width=0.97\columnwidth]{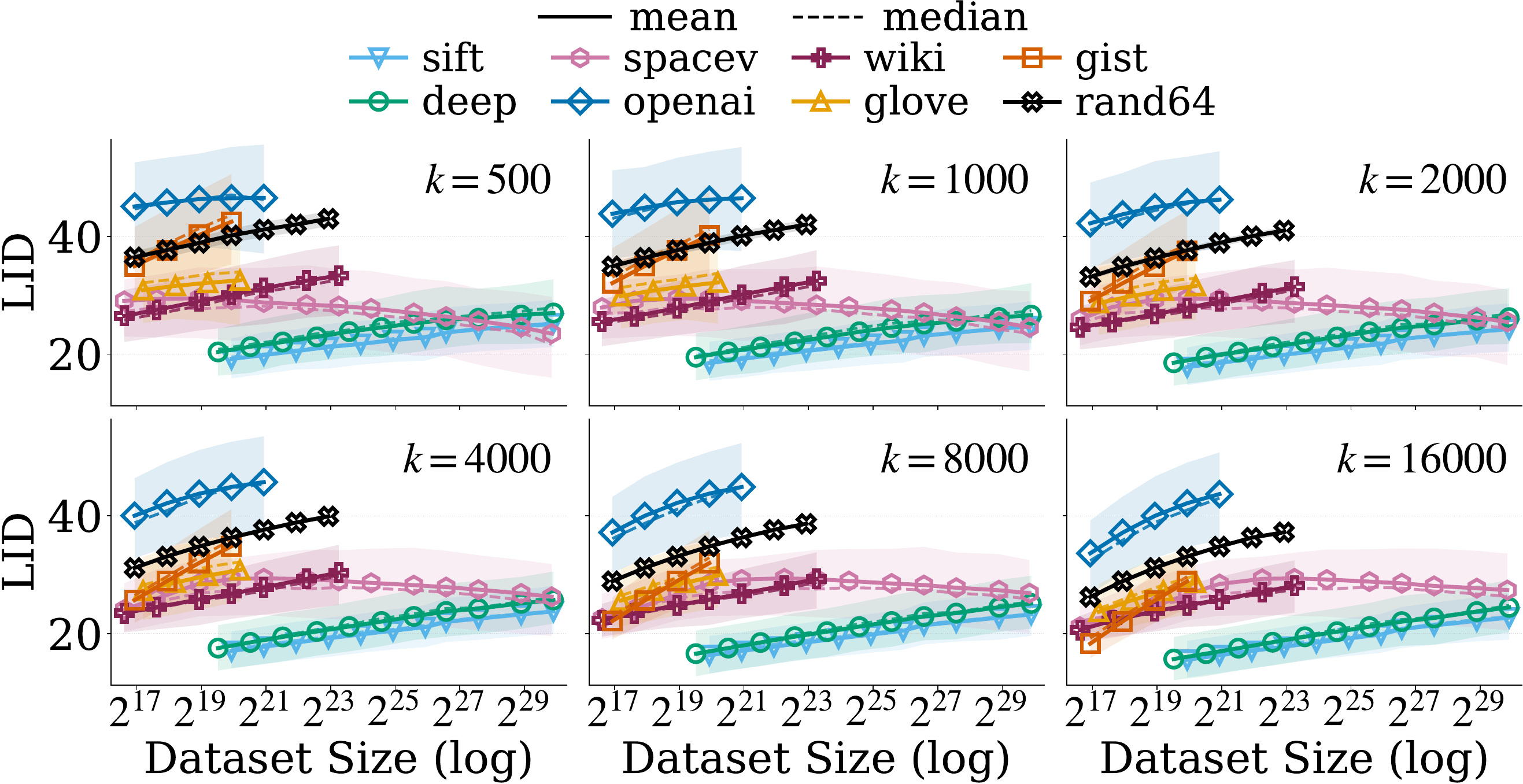}
    \Description{Line chart of median local intrinsic dimensionality versus dataset size on a logarithmic axis, one curve per dataset, each with a shaded band marking the 25th to 75th percentiles. Most curves rise roughly linearly with the logarithm of dataset size, while two curves flatten at larger sizes.}
    \caption{%
LID grows roughly with $\log N$, as the sparse regime predicts. Two datasets deviate and move toward the dense regime. LID is measured over 500 random queries and six values of $k$; shaded bands mark the 25th--75th percentiles.
}
    \label{fig:lid_main}
\end{figure}

\Paragraph{Sparse and Dense Regimes} 
In principle, $d_{\mathrm{int}}$ can take any value from $1$ to $O(\log N)$. The upper bound is motivated by the Johnson-Lindenstrauss (JL) lemma~\cite{JLREF84,dasgupta2003elementary}. The lemma states that any~$N$ points in Euclidean space can be embedded into a space with an ambient dimensionality of $O(\log N)$ while changing every pairwise distance by at most a constant factor. Since $d_{\mathrm{int}}$ never exceeds $d_{\mathrm{amb}}$~\cite{dasgupta2008random} and a constant-factor change in pairwise distances changes $d_{\mathrm{int}}$ by at most a constant factor~\cite{heinonen2001lectures, assouad1983plongements}, the same bound holds for the original dataset.\footnote{This holds even without the JL embedding: for any set of~$N$ points, every ball is covered by~$N$ half-radius balls centered at the points, so the doubling dimension, a standard notion of intrinsic dimensionality~\cite{gupta2003bounded}, is $O(\log N)$ (Theorem~\ref{thm:lg-scale-ddim}).} 
Motivated by this upper bound on $d_{\mathrm{int}}$, and guided by the LID curves in Figure~\ref{fig:lid_main}, we define two regimes of interest. In the \emph{sparse regime}, $d_{\mathrm{int}}$ grows proportionally to $\log N$, i.e., $d_{\mathrm{int}}=\Theta(\log N)$. In the \emph{dense regime}, $d_{\mathrm{int}}$ grows more slowly than $\log N$, i.e., $d_{\mathrm{int}}=o(\log N)$.
Later, in Discussion~\ref{disc:sparse-or-dense}, we further substantiate these two regimes and analyze the transition between the two. We now apply Corollary~\ref{cor:local-hyperball-scaling-n} to analyze the search cost in each regime as~$N$ grows.

\Paragraph{Sparse Regime: Power-Law}
In the sparse regime, $d_{\mathrm{int}}=\Theta(\log N)$, so $d_{\mathrm{int}}=(\log N)/\rho$ for a constant $\rho=\Theta(1)$. Plugging this into Corollary~\ref{cor:local-hyperball-scaling-n}, the expected number of points in the outer hyper-ball~is
\[
C((1+\varepsilon)r_k)
\approx
kN^{\frac{d_{\mathrm{int}}}{\log N}\log(1+\varepsilon)}
=
kN^{\log(1+\varepsilon)/\rho}
\propto
N^{c}.
\]
Notably, $c=\log(1+\varepsilon)/\rho$ is a constant independent of~$N$, which explains the power-law scaling of search cost in the sparse regime.

\Paragraph{Dense Regime: Subpolynomial}
In the dense regime, $d_{\mathrm{int}}=o(\log N)$. Plugging this into Corollary~\ref{cor:local-hyperball-scaling-n}, the expected number of points in the outer hyper-ball is
\[
C((1+\varepsilon)r_k)
\approx
kN^{\frac{d_{\mathrm{int}}}{\log N}\log(1+\varepsilon)}
=
kN^{o(1)}.
\]
The exponent $\frac{d_{\mathrm{int}}}{\log N}\log(1+\varepsilon)$ converges to 0 as $N\to\infty$, so the number of points in the outer hyper-ball grows slower than any power of~$N$, resulting in subpolynomial scaling.

\begin{figure}[t]
    \centering
    \includegraphics[width=1\columnwidth]{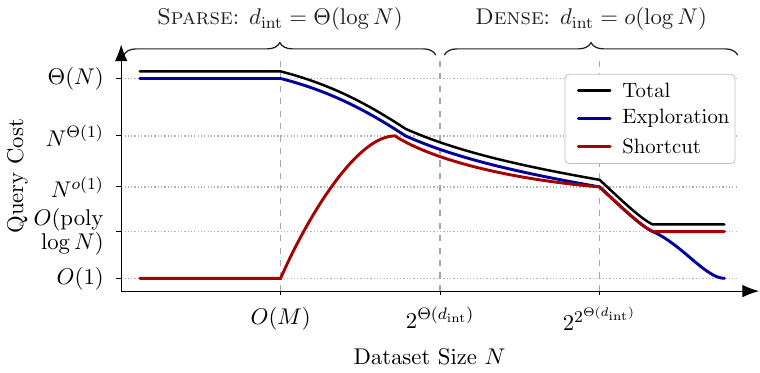}
    \Description{Schematic plot of query cost versus dataset size, partitioned into successive regimes. The growth order of the cost decreases from linear at small sizes, through Sublinear Power Laws, to poly-logarithmic at very large sizes.}
    \caption{The total query cost falls from $\Theta(N)$ to $\Theta(\text{poly}\log N)$ as $N$ crosses the regimes of Table~\ref{tab:regime-costs-uniform}.}
    \label{fig:cost}
\end{figure}

\Paragraph{Shortcut Phase} Although the exploration phase empirically accounts for most of the search cost, we analyze the shortcut phase to complete the unifying theory. Unlike the exploration phase, the shortcut phase stretches beyond the neighborhood of the query. Its cost therefore depends also on the global geometry of the dataset. Because of that, the shortcut phase's cost analysis assumes the data is uniformly distributed on a sphere. Our analysis is presented in Table~\ref{tab:regime-costs-uniform}'s first row and the formal proof of our arguments resides in Section~\ref{sec:theory-power-law}. When the dataset has fewer points than the graph out-degree ($N = O(M)$), the outer hyper-ball contains the entire dataset, so the search skips straight to the exploration phase and the shortcut phase costs $O(1)$. When the dataset is larger than the graph out-degree but has not yet saturated the underlying manifold (we later establish the saturation point is $2^{2^{\Theta(d_{\mathrm{int}})}}$), the cost of the shortcut phase grows no faster than that of the exploration phase, as seen in the table. Once the dataset saturates the manifold, the shortcut phase costs $\Theta(\text{poly}\log N)$, matching the guarantees established for fixed dimensionality~\cite{DiskANN,lin2019graph,prokhorenkova2020graph}.

\Paragraph{Unifying Theory} We now unify the two search phases into a single account of total search cost for data uniformly distributed on a sphere, summarized in Table~\ref{tab:regime-costs-uniform} and visualized in Figure~\ref{fig:cost}. While $N = O(M)$, the outer hyper-ball contains the entire dataset and the search cost is $\Theta(N)$ (Corollary~\ref{cor:local-hyperball-scaling-n}). As the dataset grows beyond $O(M)$ points but stays at most exponential in the intrinsic dimensionality ($N = 2^{O(d_{\mathrm{int}})}$), the cost follows a power law $N^c$ for $0 < c < 1$. Once the dataset exceeds $2^{O(d_{\mathrm{int}})}$ points, the cost becomes $N^{o(1)}$. Beyond $2^{2^{\Theta(d_{\mathrm{int}})}}$ points, the outer hyper-ball contains only $O(\text{poly}\log N)$ points, the shortcut phase therefore governs the scaling, and the cost is $\Theta(\text{poly}\log N)$.

\DiscussionPara{disc:sparse-or-dense}{Doubling Dimension}
The literature offers several estimators for~$d_{\mathrm{int}}$, including the doubling dimension, the expansion dimension, and the local intrinsic dimensionality~\cite{houle2017local,aumuller2021lid,assouad1983plongements}.\footnote{These measures differ in general, but for smooth data distributions they often coincide~\cite{houle2017local}.} We have already used the local intrinsic dimensionality to identify which regime the datasets of Table~\ref{table:datasets} fall into (Figure~\ref{fig:lid_main}). We now turn to the doubling dimension, prevalent in theoretical work~\cite{indyk2023worst, gupta2003bounded}. Unlike LID, which models the dataset as i.i.d.\ draws from an underlying smooth distribution, the doubling dimension is defined for any dataset as given, so guarantees proven in terms of it hold for a broader class of inputs. 

Consider a hyper-ball centered at any dataset point with any radius~$R$. Find the smallest number of non-empty hyper-balls of radius~$R/2$ that can cover all the data points within this hyper-ball. The largest such value over all dataset points and radii~$R$ is the \emph{doubling constant}~$\lambda$, and the doubling dimension is defined as $\log_2\lambda$. Computing the value of $\lambda$ exactly is NP-hard~\cite{gottlieb2013proximity}, so we use the doubling dimension only in our theoretical analysis. Theorem~\ref{thm:lg-scale-ddim} places the data in the sparse regime whenever~$N$ is at most exponential in the dimensionality of the underlying manifold.

\begin{figure}[t]
    \centering
    \includegraphics[width=0.99\columnwidth]{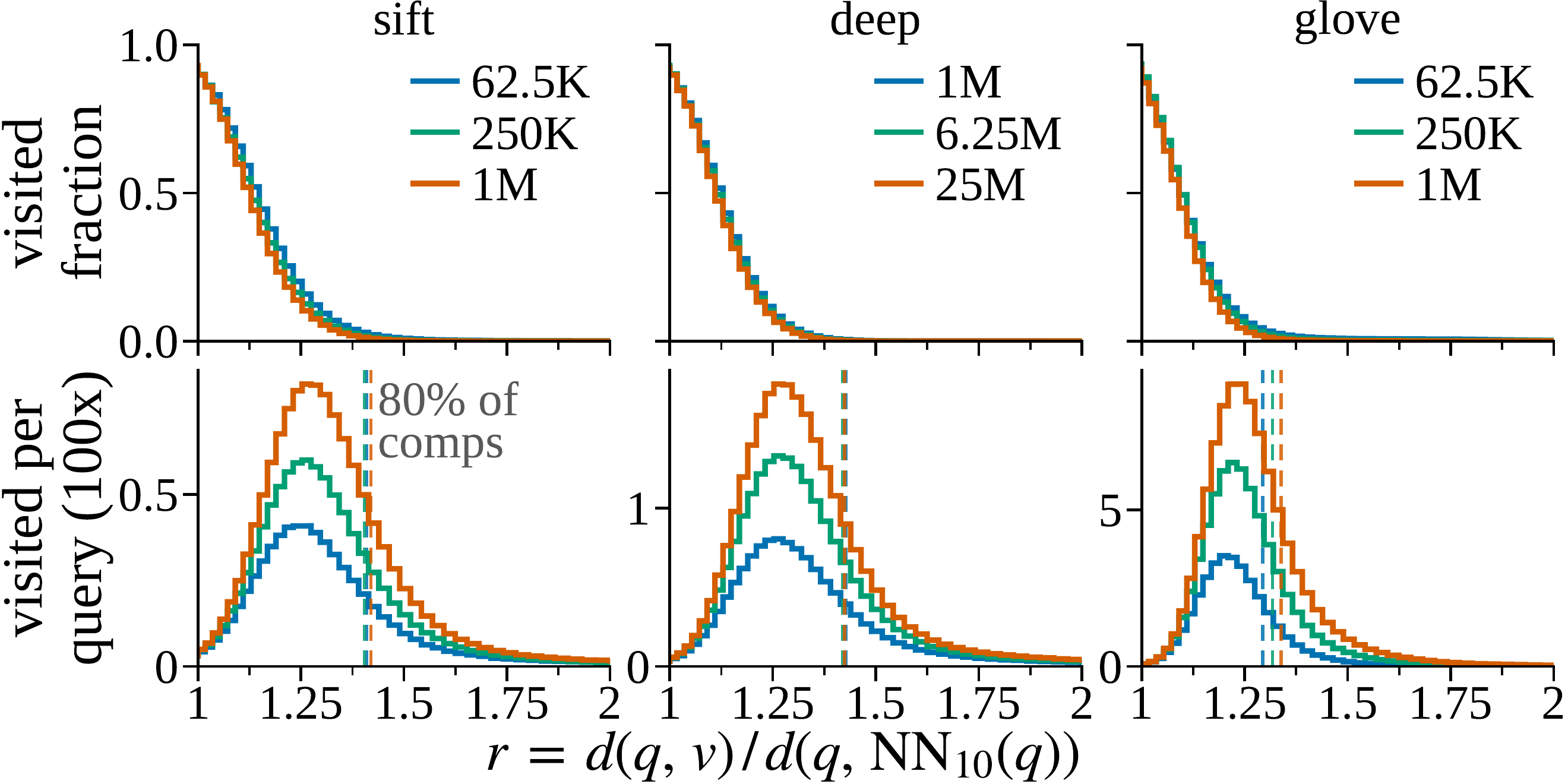}
    \Description{Two rows of histograms of normalized visited-node distances for three datasets, with one curve per dataset size. Top row: the fraction of dataset points visited per distance bucket, with the curves for different sizes nearly coinciding. Bottom row: the absolute number of distance computations per bucket, concentrated just beyond normalized distance one, with a dashed vertical line marking the distance below which 80 percent of the computations occur.}
    \caption{%
Visited-node distance distributions for HNSW across three datasets and three dataset sizes at recall $=~0.95$ ($M=32$, $efc=200$, $k=10$). Distances are normalized by the 10th-NN distance. Dashed lines indicate the normalized distance containing 80\% of distance computations.
    }
    \label{fig:visited_node_distr}
\end{figure}

\begin{theorem}[\textsc{Logarithmic Scaling of the Doubling Dimension}]
    \label{thm:lg-scale-ddim}
    Let $P$ consist of $N=2^{O(d)}$ points drawn i.i.d.\ uniformly from $\mathbb{S}^{d-1}$. With high probability, the doubling dimension of $P$ is $\Theta(\log N)$.
\end{theorem}

\begin{proof}
Fix a constant $c$ with $N\le 2^{cd}$ and set $\varepsilon=2^{-(c+1)}\le\tfrac12$. Call a pair of points \emph{bad} if its distance is at most $\varepsilon$. A fixed pair is bad only if the second point lands in a spherical cap of radius $\varepsilon$ around the first, which occupies at most an $\varepsilon^{d-2}$ fraction of the sphere.
The expected number of bad pairs is thus at most $N^2\varepsilon^{d-2}\le N\cdot 2^{cd}\cdot 2^{-(c+1)(d-2)}=2^{2c+2-d}N$, so by Markov's inequality, with probability at least $1-2^{2c+3-d}=1-2^{-\Omega(d)}$ there are fewer than $N/2$ bad pairs. In that case, remove one point of each bad pair and let $S$ be the set of remaining points; then $|S|>N/2$, and any two points of $S$ are at distance greater than $\varepsilon$.

Every ball can be covered by $\lambda$ balls of half its radius. Start with a single ball of radius $2$, which covers all of $P$, and halve $k$ times; this covers $P$ by at most $\lambda^{k}$ balls of radius $2^{1-k}$. After $k=\lceil\log_2(4/\varepsilon)\rceil\le c+4$ halvings, each ball has diameter at most $\varepsilon$, so it contains at most one point of $S$. Since all of $S$ must be covered, $\lambda^{k}\ge|S|>N/2$, hence $\dim(P)\ge(\log_2 N-1)/(c+4)=\Omega(\log N)$. Also, any ball is covered by the~$N$ balls of half its radius centered at the points, so $\lambda\le N$ and $\dim(P) = O(\log N)$.
\end{proof}

\Paragraph{Regime Change} As long as the manifold close to a query point is undersampled, each new point can locally reveal a new geometric direction of the manifold. The growing intrinsic dimensionality of Theorem~\ref{thm:lg-scale-ddim} therefore reflects newly added points progressively resolving the manifold. Once~$N$ is large enough to sample every neighborhood densely, new points reveal no new direction, and the intrinsic dimensionality stabilizes. Theorem~\ref{thm:lg-scale-ddim} also suggests where the transition from the sparse to the dense regime occurs: the logarithmic growth of $d_{\mathrm{int}}$ only holds while~$N$ is at most exponential in the manifold's dimensionality, and beyond that scale the data enters the dense regime. This aligns with the results of prior work~\cite{niyogi2008finding}, which shows that resolving the structure of a manifold requires dataset size exponential in the manifold's dimensionality.

\DiscussionPara{disc:eps}{$\varepsilon$ Interpretation}
As noted earlier in this section, growing the dataset rescales the inner and outer hyper-balls by the same factor. The value of $\varepsilon$ should therefore stay stable as the dataset grows. We report whether this holds on three datasets from Table~\ref{table:datasets}; the others exhibit a similar pattern. For each dataset we record every node $v$ that the beam search visits, i.e., every node for which it computes the distance to the query. For each visited node we then compute the normalized distance to the query point: $d(q,v)/r_{k}$ with $k=10$. Figure~\ref{fig:visited_node_distr} shows the histogram of these normalized distances and reveals how far beyond the true nearest neighbors the search explores to find them.

In the top row, in each bucket of the normalized distance, we report the fraction of dataset points that the search process visits. For each dataset, this profile nearly coincides across dataset sizes. The fraction of points the search visits at each normalized distance is therefore scale-invariant. As a result, the outer hyper-ball rescales similarly with the inner hyper-ball, and $\varepsilon$ stays stable as the dataset~grows.

The bottom row shows where the search process spends most of its effort, reporting the absolute number of visited nodes, i.e., distance computations, per bucket. At every dataset size, the cost concentrates in a band of near-ties just beyond the inner hyper-ball. The vertical dashed line marks the normalized distance below which $80\%$ of the distance computations occur, and this threshold stays small at every scale.

\DiscussionPara{disc:assumptions}{Relaxing the Assumptions}
We now revisit the three assumptions behind our analysis and examine what changes when they are relaxed. The first assumption is that queries are drawn from the same distribution as the data. When this does not hold, as in out-of-distribution workloads, the outer hyper-ball may intersect multiple manifolds~\cite{brown2022verifying}. In this case, the growth of the exploration phase's cost is governed by the largest intrinsic dimensionality among the intersected manifolds. 

The second assumption is that the dataset is sufficiently dense. This ensures that the outer hyper-ball remains small relative to the scale at which the manifold bends. When the dataset is sparse, however, the projection of the hyper-ball onto the manifold may no longer resemble a hyper-ball and can instead take more irregular shapes. This changes the number of points contained in the hyper-ball and, consequently, affects Corollary~\ref{cor:local-hyperball-scaling-n}. Nevertheless, the corollary continues to hold as long as the intrinsic dimensionality realized at each scale follows the same trend across the sparse and dense regimes. The same argument applies when the local uniformity model is violated: our analysis remains valid as long as this trend in intrinsic dimensionality is preserved.\footnote{The maximum-likelihood estimator of local intrinsic dimensionality assumes a smooth underlying distribution~\cite{levina2004maximum}, which need not hold in out-of-distribution cases or when the dataset is too sparse.} 

Finally, we assume that the distance metric is Euclidean. If an angular metric is used instead, Euclidean distance between $\ell_2$-normalized vectors is a strictly increasing function of the angle, so the two metrics induce the same ordering of distances from any query point and, in particular, the same nearest neighbors~\cite{qian2004similarity, jain2016approximate}. Even more strongly, the two metrics are locally essentially identical; the outer hyper-ball around the query is therefore unaffected by the choice of metric, and the analysis carries over without modification.

\subsection{Theoretical Analysis}
\label{sec:theory-power-law}

In Section~\ref{sec:anatomy_of_query_cost}, we estimated the search cost of graph-based indexes by empirically tying the search cost to the number of points in the query's neighborhood. In this section, we rigorously establish the same cost bounds by also accounting for the effect of these graphs' edges on the search. We analyze the costs of the shortcut and exploration phases in the sparse regime (Section~\ref{sec:theory-sparse}), followed by the dense regime (Section~\ref{sec:theory-dense}). Together, these analyses establish the bounds of Table~\ref{tab:regime-costs-uniform} and the relevant transition thresholds. 

\Paragraph{Setting}
To streamline our exposition, we focus on SNG~\cite{SNG} in our proofs, and later argue how the same proofs can be adapted to other graph-based indexes such as HNSW~\cite{HNSW} and Vamana~\cite{DiskANN}. To provide the strongest result possible, we study the query cost of the exact version of SNG and allow the construction \mbox{to run free of charge}. 
We also discuss how the same results hold for the practical versions of the graph-based indexes in the sparse regime, as long as they are allowed to compute the highest quality neighbors for each node using a candidate list of unbounded size (i.e., $ef_{\text{construction}} = \infty$).

We analyze the exploration phase under the local uniformity model, where the points within a local neighborhood of a query are approximately uniformly distributed. Since the shortcut phase depends on the global geometry of the data, one cannot prove a non-trivial bound for it in the local uniformity model without additional assumptions. Thus, we restrict the shortcut phase's analysis to datasets whose points are uniformly distributed on a sphere. Henceforth, we refer to the costs of these two phases in their respective settings as the \emph{locally uniform exploration} and \emph{uniform shortcut} costs.

\subsubsection{\textbf{Sparse Regime}}
\label{sec:theory-sparse}
In the sparse regime, $d_{\mathrm{int}} = \Theta(\log N)$. As discussed in Section~\ref{sec:anatomy_of_query_cost}, such a value of~$d_{\mathrm{int}}$ causes the number of points within the approximately uniform neighborhood of a query point to be~$N^c$, where~$N$ is the dataset size and~$c \in (0,1]$ is a constant. We now prove that SNG's locally uniform exploration cost grows as a power law in the size of such a neighborhood, yielding a Sublinear Power Law in~$N$. We then show that the same argument implies that SNG's uniform shortcut cost grows similarly.

\Paragraph{Reduction to~$M$-NN Graphs or High-Degree Graphs}
The central part of our argument shows that, when using SNG to index a set of uniformly distributed points, one of two scenarios occurs: (1)~if the maximum degree~$M$ is small, SNG connects each node to its~$M$ nearest neighbors with high probability, and (2)~if~$M$ is large, the average node degree follows a power law in~$N$. In the former case, SNG behaves similarly to an~$M$-NN graph, which is known to have a power-law search cost~\cite{laarhoven2017graph,prokhorenkova2020graph}. In the latter case, navigating a single node entails computing the distances of a power-law number of neighboring nodes to the query. Together, these cases imply that the uniform shortcut cost and the locally uniform exploration cost follow Sublinear Power Laws in~$N$, giving rise to the leftmost two columns of Table~\ref{tab:regime-costs-uniform}.

\Paragraph{Connection Probability}
To prove the above results, we begin by bounding the probability of each node connecting to its~$k${\nobreakdash-}th nearest neighbor:
\begin{lemma}[Connection Probability]
    \label{lemma:connection_probability}
    Consider a node~$u$ within a region of the space where points are (approximately) uniformly distributed. Denoting~$u$'s $k$-th nearest neighbor as~$v_k$ and letting $p^*=d_{\mathrm{int}}^{\Theta(1)} \cdot 2^{-d_{\mathrm{int}}/2}=\log^{\Theta(1)} N \cdot N^{-c/2}$, $u$ is connected to~$v_k$ with a probability of at least~$1 - 2 \cdot (k-1) \cdot p^*$ when~$k \leq M$.
\end{lemma}
\begin{proof}
    We observe that the exact version of SNG connects~$u$ to its~$M$ nearest neighbors that are not pruned. Recall that a node~$v_k$ is pruned if~$u$ is already connected to some other node~$v_i$ with $i < k$ that satisfies~$d(u,v_k) \geq d(v_i,v_k)$. This corresponds to the neighbor~$v_i$ being included in the ``lune'' between~$u$ and $v_k$, as shown in Figure~\ref{fig:graph_constructions}{\nobreakdash-}b.
    
    Thus, denoting the lune between~$u$ and $v_k$ by $L_k$, $v_k$ is pruned with a probability of at most
    \begin{align*}
        \Pr[v_k \text{ pruned}] & = \Pr[\exists v_i, i < k: v_i \in L_k] \\
        & \leq \sum_{i=1}^{k-1} \Pr[u \text{ connected to } v_i \land v_i \in L_k] \\
        & \leq \sum_{i=1}^{k-1} \Pr[v_i \in L_k]. \\
    \end{align*}
    Now, we observe that the lune~$L_k$ is comprised of two spherical caps of height~$d(u,v_k)/2$. The probability of a uniformly distributed point~$v_i$ with $i < k$ being within such a spherical cap is at most the ratio of the volume of this cap to the volume of the hyper-ball of radius~$d(u,v_k)$ around $u$. This quantity is known to be~$p^* = d_{\text{int}}^{\Theta(1)} \cdot 2^{-d_{\text{int}}/2}$ \cite[Lemma~2]{laarhoven2017graph}. Thus, the probability that~$v_i$ is within the lune~$L_k$ is at most~$2 \cdot p^*$. Plugging this expression into the sum above yields $\Pr[v_k \text{ pruned}] \leq \sum_{i=1}^{k-1} 2 \cdot p^* = 2 \cdot (k-1) \cdot p^*$. Finally, noting that~$u$ must be connected to $v_k$ if $v_k$ is not pruned yields 
    $$ 
    \Pr[u \text{ connected to } v_k] = 1-\Pr[v_k \text{ pruned}] \geq 1 - 2 \cdot (k-1) \cdot p^*,
    $$
    establishing the lemma.
\end{proof}

%\Paragraph{Consequences of Lemma~\ref{lemma:connection_probability}}
Lemma~\ref{lemma:connection_probability} has two key consequences in the following two cases of the maximum degree~$M$:

\textbf{Case (1) of Small~$M$.}
This case leads to a large subgraph of SNG being an~$M${\nobreakdash-}NN graph with high probability within a region of the space where the dataset is approximately uniformly distributed:
\begin{theorem}[\textsc{$M$-NN Subgraph}]
    \label{thm:case_1}
    Consider a point~$x$ around which the dataset is approximately uniformly distributed. The~$n = 1/(2 \cdot M^2 \cdot p^*) = N^{c/2} / (2 \cdot M^2 \cdot \log^{\Theta(1)}N)$ nearest points to~$x$ are connected to their~$M$ nearest neighbors with a probability of at least~$1/2$.
\end{theorem}
\begin{proof}
    By applying a union bound to the connection probability of Lemma~\ref{lemma:connection_probability}, it follows that SNG connects a node~$u$ to all of its~$M$ nearest neighbors with a probability of at least~$1-M^2 \cdot p^* = 1 - M^2 \cdot \log^{\Theta(1)} N \cdot N^{-c/2}$. Through another union bound, this implies that all of the~$n = N^{c/2} / (2 \cdot M^2 \cdot \log^{\Theta(1)} N)$ nearest points to~$x$ must be connected to their~$M$ nearest neighbors with a probability of at least~$1/2$.
\end{proof}
When~$M = N^{o(1)}$, Theorem~\ref{thm:case_1} implies that in the subgraph of SNG restricted to the~$N^{c/2 - o(1)}$ nearest points of~$x$, denoted by~$\overline{G}$, each node is connected to at most~$M$ of its nearest neighbors. Since the points in~$\overline{G}$ are uniformly distributed, this intuitively implies that each of its nodes is connected to at most~$M$ other nodes that are ``random,'' meaning that they may or may not aid in navigating closer to a query. In such a scenario, the search must traverse many nodes and compare the distances of many points to the query to find its nearest neighbor with a non-trivial probability.

This intuition is formalized in~\cite{laarhoven2017graph}, where the author shows that using greedy search on an~$M${\nobreakdash-}NN graph\footnote{More precisely, the graph the author of~\cite{laarhoven2017graph} analyzes connects any two nodes that are within a tunable distance threshold. This thresholding causes the nodes of this graph to have variable degrees rather than strictly bounded degrees. Nevertheless, the proofs of~\cite{laarhoven2017graph} can be adapted to the setting where the degree is strictly bounded by~$M$.} to reach a node that is at most a constant factor farther than a query's true nearest neighbor entails a number of distance computations that follows a Sublinear Power Law in the graph's size.\footnote{Specifically, instead of classic beam search, \cite{laarhoven2017graph} analyzes a query algorithm where the search procedure restarts from a random node when it converges to a local minimum. Nevertheless, one can analyze beam search using a similar approach to random restarts.} A slight adaptation of their proof would extend their result to any graph where each node is connected to some number of its nearest neighbors that is less than or equal to~$M$. 

To analyze SNG's uniform shortcut cost, we consider the subgraph~$\overline{G}_{\text{entry}}$ around SNG's entry point and apply Theorem~\ref{thm:case_1}. Then, by invoking the result of~\cite{laarhoven2017graph} for~$\overline{G}_{\text{entry}}$, we derive the cost as a power law in the size of the entry point's neighborhood, i.e., $N^c$ for some constant $c \in (0, 1]$. As discussed at the beginning of this section, this corresponds to a Sublinear Power Law in~$N$. Similarly, to analyze SNG's locally uniform exploration cost, we consider the subgraph~$\overline{G}_q$ around the query~$q$ and leverage the approximately uniform distribution of its points to apply Theorem~\ref{thm:case_1}. This allows us to once again invoke the result of~\cite{laarhoven2017graph} to show that the cost follows a power law in~$N$.

\textbf{Case (2) of Large~$M$.}
In this case, each node within the approximately uniform neighborhood of the query~$q$ has a high degree with high probability:
\begin{theorem}[\textsc{High-Degree Graph}]
    \label{thm:case_2}
    When the dataset is uniformly distributed and~$M=N^{\Theta(1)}$, each node within the uniform neighborhood of a query~$q$ has a degree of at least~$\min(M, N^{c/4-o(1)})=N^{\Theta(1)}$ with a probability of at least~$1/2$.
\end{theorem}
\begin{proof}
    If~$M = N^{\Theta(1)} < N^{c/4}/(2 \cdot \log^{\Theta(1)} N)$, a similar union bound to that of Theorem~\ref{thm:case_1} applied to the connection probability of Lemma~\ref{lemma:connection_probability} shows that any node~$u$ will be connected to its~$M=N^{\Theta(1)}$ nearest neighbors with a probability of at least~$1/2$. If~$M \geq N^{c/4}/(2 \cdot \log^{\Theta(1)} N)$, by the same union bound, any node $u$ will be connected to its~$N^{c/4}/(2 \cdot \log^{\Theta(1)} N)$ nearest neighbors with a probability of at least~$1/2$.
\end{proof}
Theorem~\ref{thm:case_2} implies that, for both the uniform shortcut cost and the locally uniform exploration cost, processing even a single node entails a number of distance computations that follows a power law in~$N$.

As such, the results in the two columns on the left of Table~\ref{tab:regime-costs-uniform} hold in both cases of~$M$. Note that the uniform shortcut cost in the leftmost column is constant since, in that setting, Theorem~\ref{thm:case_1} implies that all nodes are connected to each other, so the search begins in the exploration phase and skips the shortcut phase.

\Paragraph{Extension to Other Graph-based Indexes}
The exact versions of the major graph-based indexes, such as HNSW, Vamana, and \(\tau\)-MG, employ pruning rules at least as strict as that of SNG. As such, an analogue of Lemma~\ref{lemma:connection_probability} holds for them with a probability bound that is at least as high. This implies that the same conclusions as those of Theorems~\ref{thm:case_1} and \ref{thm:case_2} hold for these graphs, as these results only depend on the bounds of Lemma~\ref{lemma:connection_probability}.

For the practical versions of these graphs, when allowed to find the highest quality neighbors of each node using an unbounded candidate list (i.e., $ef_\text{construction}=\infty$), the same Sublinear Power Law results hold.\footnote{Here, we consider incremental insertion as most practical graph indexes are implemented accordingly.} The reason is two-fold: (1)~When the maximum degree~$M$ is small, the same analysis as Lemma~\ref{lemma:connection_probability} shows that, with high probability, each node after the first~$M$ nodes is connected to the~$M$ nearest nodes from those that were inserted earlier. This does not affect our use of the results in~\cite{laarhoven2017graph}, as it only slightly alters the probability of the search encountering a local minimum. (2)~When~$M$ is large, Lemma~\ref{lemma:connection_probability} applied to the~$N/2$ points inserted later during construction implies that their degrees must be asymptotically as large as that of Theorem~\ref{thm:case_2}. Thus, processing even one of these nodes, which happens with a constant probability, incurs a power-law cost.

\subsubsection{\textbf{Dense Regime}}
\label{sec:theory-dense}
We now turn to the dense regime and analyze the uniform shortcut and locally uniform exploration costs. We begin by showing that the node degree is negligible in this regime, which leaves the bulk of the search costs to the number of steps. We lower bound the number of shortcut steps by analyzing the average progress each step makes towards the query. For the locally uniform exploration cost, we simply bound the number of points within a query's neighborhood.

\Paragraph{Node Degree}
As proven in \cite[Lemma~3.3]{indyk2023worst}, the degree of a node in Vamana is at most~$O(2^{\Theta(d_{\text{int}})} \cdot \text{poly}\log N)$.\footnote{The poly-logarithmic term corresponds to $\log \Delta$ in~\cite[Lemma~3.3]{indyk2023worst}, where~$\Delta$ is the ratio between the longest and shortest pairwise distances of points in the dataset. We have assumed that $\Delta=\text{poly}(N)$ to reflect real-world datasets.} We further observe that applying the same lemma with a value of~$\alpha=1$ shows that the degree of each node in SNG or HNSW is at most~$O(2^{\Theta(d_{\text{int}})})$. Since, as described in Section~\ref{sec:anatomy_of_query_cost}, $d_{\text{int}} = o(\log N)$ in the dense regime, we get the following lemma:
\begin{lemma}[Negligible Degree]
    \label{lemma:negligible_degree}
    In the dense regime, the degree of any node in SNG, HNSW, or Vamana is at most~$O(2^{\Theta(d_{\text{int}})} \cdot \text{poly}\log N) = N^{o(\log N)/\log N} \cdot \text{poly}\log N = N^{o(1)}$.
\end{lemma}
Lemma~\ref{lemma:negligible_degree} posits that processing a single node, be it in the shortcut phase or the exploration phase, takes~$N^{o(1)}$ time. In particular, when~$N \geq 2^{2^{d_{\text{int}}}}$, we have $d_{\text{int}} = O(\log\log N)$, which implies that the degree of each node is at most $O(2^{O(\log \log N)}) = O(\text{poly}\log N)$.  This establishes that the uniform shortcut and locally uniform exploration costs are at least as large as the expressions in the right two columns of Table~\ref{tab:regime-costs-uniform}. We now discuss upper bounds on these costs by analyzing the number of nodes processed in each phase.

\Paragraph{Uniform Shortcut Cost}
For Vamana, it is known that each step of the shortcut phase takes the search closer to any query~$q$ by a constant factor~\cite{indyk2023worst}. From \cite[Lemma~2]{laarhoven2017graph}, we observe that the nearest neighbor of any query~$q$ is within an expected distance of~$N^{-\Theta(1/d_{\mathrm{int}})}$ of it, implying that the outer hyper-ball around~$q$ must have a radius of~$O(N^{-\Theta(1/d_{\mathrm{int}})})$. Thus, since the distance from the graph's entry point to~$q$ is at most a constant when the points are distributed on a sphere, the number of shortcut steps is at most~$O(\log \frac{O(1)}{O(N^{-\Theta(1/d_{\mathrm{int}})})}) = \Theta((1/d_{\mathrm{int}}) \cdot \text{poly}\log N)$. Multiplying by each node's degree, one derives upper bounds matching the expressions in the two columns on the right of Table~\ref{tab:regime-costs-uniform}.

Compared to Vamana, upper bounding SNG and HNSW's uniform shortcut costs is more complex. In fact, we show the number of steps in SNG's shortcut phase is at least~$\Omega(N^{-1/d_{\mathrm{int}}})$, i.e., significantly higher than Vamana's poly-logarithmic number of steps. Our argument quantifies the average progress each search step makes:
\begin{lemma}[SNG Maximum Progress]
    \label{lemma:sng_expected_progress}
    In SNG's shortcut phase on a dataset of~$N$ points uniformly distributed on a sphere, with high probability, all steps that take the search closer to a query~$q$ reduce the distance to~$q$ by at most~$\text{poly}(d_{\mathrm{int}}) \cdot \left( \frac{\log N}{N} \right)^{1/d_{\mathrm{int}}}$.
\end{lemma}
\begin{proof}
    Our argument analyzes the maximum length of an edge in SNG. Recall that if two nodes~$u$ and $v$ are connected, then there must be no other node in their lune. By \cite[Lemma~2]{laarhoven2017graph}, the probability that a node~$w$ is within~$u$ and $v$'s lune is~$\text{poly}(d_{\mathrm{int}}) \cdot d(u,v)^{d_{\mathrm{int}}}$. When~$d(u,v) = \left( \frac{\log N}{N} \right)^{1/d_{\mathrm{int}}}$, this probability becomes $\text{poly}(d_{\mathrm{int}}) \cdot \frac{\log N}{N}$, i.e., the expected number of nodes within the lune is~$\Omega(\log N)$. By applying a standard Chernoff bound, it follows that such a lune is empty with a probability of at most~$1/N^3$. Thus, by a union bound over all pairs of nodes, we can conclude that all edges within SNG have a length of at most $\text{poly}(d_{\mathrm{int}}) \cdot \left( \frac{\log N}{N} \right)^{1/d_{\mathrm{int}}}$ with high probability. Therefore, any edge that brings the search closer to a query must also be at most this long.
\end{proof}
As discussed above, the radius of a query's outer hyper-ball is~$O(N^{-1/d_{\mathrm{int}}})$ in expectation. Thus, since the distance from the entry point to a query can be a constant number, Lemma~\ref{lemma:sng_expected_progress} implies that the shortcut phase requires at least~$\frac{O(1)}{\text{poly}(d_{\mathrm{int}}) \cdot (\log N / N)^{1/d_{\mathrm{int}}}} = \Omega(N^{(1-o(1))/d_{\mathrm{int}}})$ many steps to reach the outer hyper-ball. The actual number of steps required to reach this hyper-ball may be larger than this quantity, since beam search may run into local minima which prevent it from making progress towards a query.

Since each layer of HNSW behaves similarly to SNG, similar concerns around an upper bound on its search cost persist. We leave deriving a tight upper bound for SNG and HNSW as an interesting open question for future work.

\Paragraph{Locally Uniform Exploration Cost}
The exploration phase in any graph-based index searches at most all nodes within the locally uniform neighborhood of a query. Thus, the number of steps in the exploration phase is upper bounded by the number of nodes in this neighborhood. Since~$d_{\text{int}} = o(\log N)$ in the dense regime, Theorem~\ref{thm:local-hyperball-scaling} implies that the number of nodes in the query's uniform neighborhood is at most~$2^{d_{\text{int}}}=N^{d_{\text{int}}/\log N}=N^{o(1)}$. Multiplied by the degree of a node (quantified in Lemma~\ref{lemma:negligible_degree}), this yields a locally uniform exploration cost of~$N^{o(1)}$, matching the third column of Table~\ref{tab:regime-costs-uniform}. Moreover, when~$N \geq 2^{2^{\Theta(d_{\text{int}})}}$, we have $d_{\text{int}} = O(\log \log N)$, making the number of nodes in the query's neighborhood poly-logarithmic in~$N$. Multiplied by the node degree (Lemma~\ref{lemma:negligible_degree}), which is poly-logarithmic in this case, we derive a poly-logarithmic cost for the exploration phase. Thus, the bound in the fourth column of Table~\ref{tab:regime-costs-uniform} also holds.

\begin{figure*}[t]
  \centering
  \includegraphics[width=1\textwidth]{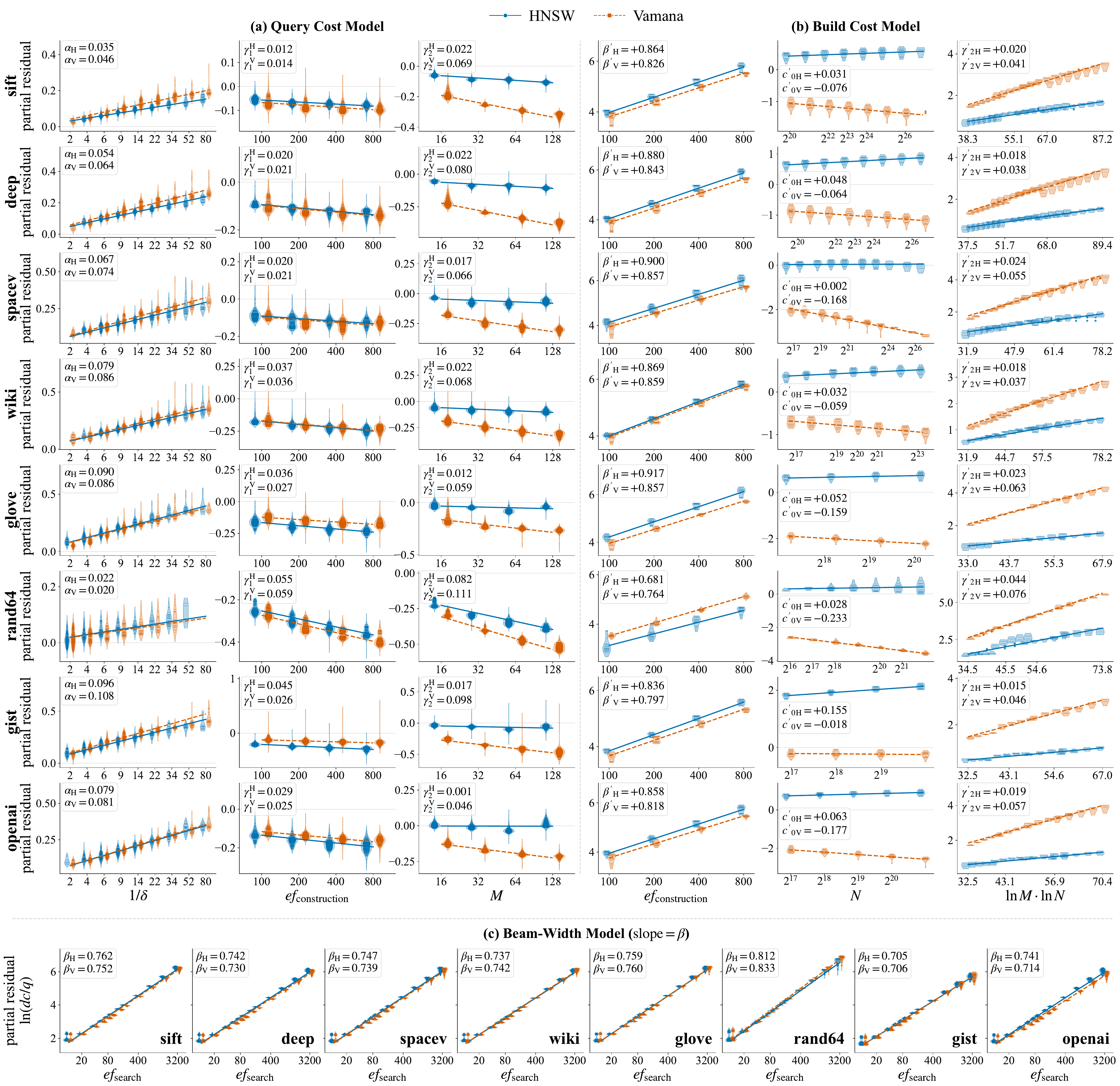}
  \Description{Three groups of partial-regression scatter plots, each with a fitted straight line. Group (a) shows the effects of recall, construction beam width, and degree bound on the query scaling exponent. Group (b) shows the effects of construction beam width, dataset size, and their interaction on average insertion cost. Group (c) plots search beam width against per-query distance computations. In all subfigures the points cluster tightly around the fitted lines.}
  \caption{
  Partial regression plots for the query-, insertion-, and beam-width-cost models; each subfigure shows the partial residual after removing all other fitted terms, so the fitted slope equals the associated model coefficient.
  \textbf{(a)}~Query-side model of Eq.~\ref{eq:c_model}: the effects of $\ln(1/\delta)$, $\ln(ef_{\text{construction}})$, and $\ln M$ on the scaling exponent~$c$.
  \textbf{(b)}~Insertion-cost model of Eq.~\ref{eq:build_model}: the effects of $\ln(ef_{\text{construction}})$, $\ln N$, and the interaction $\ln(M)\ln(N)$ on average insertion cost $\ln(\mathrm{dc}_{\text{build}}/N)$.
  \textbf{(c)}~Beam-width model of Eq.~\ref{eq:beta}: $\ln(ef_{\text{search}})$ vs.\ $\ln(\mathrm{dc}/q)$, centered within each configuration; the linear pattern confirms the power law with slope $\beta$.
  }
  \label{fig:cost_models}
\end{figure*}

\section{Cost Modeling}
\label{sec:cost-model}

The previous sections established that, across many datasets and scales, query cost follows a Sublinear Power Law in the dataset size~$N$, with exponent \(c\), and explained why. A natural next step is to predict the exponent, so that its scaling is known before the data grows. This section develops such models empirically for query cost (Section~\ref{sec:query-cost-model}) and insertion cost (Section~\ref{sec:insert_cost}). Together, they cover the parameter ranges of practical interest and expose the trade-offs between query cost, insertion cost, and recall.

\subsection{Modeling Query Cost}
\label{sec:query-cost-model}

We begin with query cost, asking what determines the exponent \(c\) for a given dataset. Our approach is to model how \(c\) varies with the recall target and the construction parameters of the index, and then to connect the predicted cost back to the $ef_{\text{search}}$ value that realizes it.

\Paragraph{Parameters} We focus on three parameters that govern how \(c\) varies for a given dataset. The first is the recall target. A higher target leaves less room to miss a true nearest neighbor, so the search process must cover more of the local neighborhood and widen its candidate list. The other two are the construction list size \(ef_{\text{construction}}\) and the maximum out-degree \(M\), which together set the quality of the graph and the number of edges it carries. Those edges act as shortcuts that bring every query within fewer hops of its neighborhood, at the cost of more work during construction.

\Paragraph{Expected Effects} Increasing recall should raise $c$, but the rise is not linear in recall. Raising recall from $0.90$ to $0.92$ barely makes ANN search harder, but raising it from $0.97$ to $0.99$ does. To capture the nonlinearity, we quantify the demand for recall as $1/\delta$, where $\delta = 1 - \text{recall}$. In the previous example, the first increment raises $1/\delta$ from $10$ to $12.5$, while the second raises it from $33.3$ to $100$. 
On the other hand, increasing the construction parameters should lower $c$, since a denser graph offers more reliable routes into and within the query's neighborhood. This benefit should diminish and eventually turn harmful: once the graph is dense enough, additional edges add little routing power while inflating the search cost per step. Across the parameter ranges we study (recall up to $0.99$, $ef_{\text{construction}} \in [100, 800]$, $M \in [16, 128]$), this reversal never appears, and we observe $c$ decreasing throughout.

\begin{table}[t]
    \scriptsize
    \centering
    \setlength{\tabcolsep}{2pt}
    \caption{Fitted parameters and goodness of fit per dataset and algorithm.
        (a) Query cost model: coefficients of Eq.~\ref{eq:c_model} and the
        exponent $\beta$ of Eq.~\ref{eq:beta}, with $R^2_c$ and $R^2_\beta$
        the respective fits. (b) Build cost model: coefficients of
        Eq.~\ref{eq:build_model}.}
    \label{tab:model_params}
    \begin{tabular}{ll ccccccc @{\hspace{8pt}\vrule\hspace{8pt}} ccccc}
        & & \multicolumn{7}{c}{\textbf{(a) Query cost model}} &
            \multicolumn{5}{c}{\textbf{(b) Build cost model}} \\
        \cmidrule(l{2pt}r{16pt}){3-9} \cmidrule(l{2pt}r{2pt}){10-14}
        & Dataset & $c_0$ & $\alpha$ & $\gamma_1$ & $\gamma_2$ & $\beta$ & $R^2_c$ & $R^2_\beta$ &
            $\kappa$ & $\beta'$ & $c'_0$ & $\gamma'_2$ & $R^2$ \\
        \midrule
        \multirow{8}{*}{\rotatebox[origin=c]{90}{\textbf{HNSW}}}
        & sift   & .244 & .035 & .012 & .022 & .762 & .919 & .992 & 2.497 & .864 & \phantom{$-$}.031 & .020 & .979 \\
        & deep   & .281 & .054 & .020 & .022 & .742 & .934 & .992 & 2.163 & .880 & \phantom{$-$}.048 & .018 & .983 \\
        & spacev & .181 & .067 & .020 & .017 & .747 & .818 & .992 & 2.653 & .900 & \phantom{$-$}.002 & .024 & .963 \\
        & wiki   & .310 & .079 & .037 & .022 & .737 & .886 & .991 & 2.424 & .869 & \phantom{$-$}.032 & .018 & .991 \\
        & glove  & .336 & .090 & .036 & .012 & .759 & .916 & .993 & 2.147 & .917 & \phantom{$-$}.052 & .023 & .977 \\
        & rand64 & 1.068 & .022 & .055 & .082 & .812 & .912 & .992 & 3.569 & .681 & \phantom{$-$}.028 & .044 & .922 \\
        & gist   & .513 & .096 & .045 & .017 & .705 & .859 & .992 & 1.268 & .836 & \phantom{$-$}.155 & .015 & .991 \\
        & openai & .191 & .079 & .029 & .001 & .741 & .910 & .992 & 2.197 & .858 & \phantom{$-$}.063 & .019 & .981 \\
        \midrule
        \multirow{8}{*}{\rotatebox[origin=c]{90}{\textbf{Vamana}}}
        & sift   & .415 & .046 & .014 & .069 & .752 & .841 & .992 & 3.622 & .826 & $-$.076 & .041 & .977 \\
        & deep   & .479 & .064 & .021 & .080 & .730 & .869 & .992 & 3.388 & .843 & $-$.064 & .038 & .973 \\
        & spacev & .351 & .074 & .021 & .066 & .739 & .825 & .991 & 4.253 & .857 & $-$.168 & .055 & .976 \\
        & wiki   & .447 & .086 & .036 & .068 & .742 & .845 & .992 & 3.137 & .859 & $-$.059 & .037 & .986 \\
        & glove  & .480 & .086 & .027 & .059 & .760 & .835 & .993 & 3.977 & .857 & $-$.159 & .063 & .990 \\
        & rand64 & 1.252 & .020 & .059 & .111 & .833 & .929 & .993 & 5.239 & .764 & $-$.233 & .076 & .996 \\
        & gist   & .670 & .108 & .026 & .098 & .706 & .782 & .991 & 2.654 & .797 & $-$.018 & .046 & .985 \\
        & openai & .334 & .081 & .025 & .046 & .714 & .947 & .991 & 4.492 & .818 & $-$.177 & .057 & .983 \\
        \bottomrule
    \end{tabular}
\end{table}

\Paragraph{Modeling the exponent~$c$}
The simplest model consistent with these effects treats them as independent: each parameter contributes its own term, and the exponent~$c$ is their sum. Figure~\ref{fig:cost_models}-a supports this decomposition empirically. Each column isolates the effect of one factor on the exponent (\(1/\delta\), \(ef_{\text{construction}}\), \(M\)). We subtract the fitted contributions of the other two, and plot what remains, the \emph{partial residual}, against the factor in question. In each case the partial residual is approximately linear when the factor is shown on a logarithmic scale, so each factor contributes an additive logarithmic term to the exponent. We therefore model the scaling exponent as
\begin{equation}
\label{eq:c_model}
    c \;=\; c_0
    \;+\; \alpha \,\ln\!\left(\tfrac{1}{\delta}\right)
    \;-\; \gamma_1 \,\ln(ef_{\text{construction}})
    \;-\; \gamma_2 \,\ln(M).
\end{equation}
Here \(c_0\) is the baseline set by the dataset's geometry, \(\alpha\) the cost of higher demand for recall, and \(\gamma_1\) and \(\gamma_2\) the benefits of \(ef_{\text{construction}}\) and \(M\), respectively.

\Paragraph{Goodness of fit} We fit the model separately for each dataset and index construction algorithm. Table~\ref{tab:model_params}-a reports the fitted coefficients and goodness of fit. The model achieves an average \(R^2\) of approximately \(0.88\), meaning it explains about \(88\%\) of the variation in the measured exponent across datasets and algorithms. The fitted coefficients are consistent with the slopes in Figure~\ref{fig:cost_models}-a, showing the additive log-linear structure directly on the data.

\Paragraph{Effect of $ef_\text{search}$} The model above explains how fast the number of distance computations scales as the dataset grows, i.e., how large the exponent \(c\) is for a given recall level and construction parameters. It leaves open a practical question: for a dataset of a given size, how should the candidate list size during search, \(ef_{\text{search}}\), be tuned so that the search performs the distance computations predicted by Equation~\ref{eq:c_model} and thus reaches the recall target? We answer this by observing that, for a given graph-based index, distance computations per query scale as a Sublinear Power Law in \(ef_{\text{search}}\):
\begin{equation}
\label{eq:beta}
    \text{dc}/q \;\propto\; ef_{\text{search}}^{\,\beta}.
\end{equation}
Figure~\ref{fig:cost_models}-c shows this relationship on log-log axes, where the linear pattern indicates a power law, this time between \(ef_{\text{search}}\) and the number of distance computations. The fit is tight across all datasets, with \(R^2 > 0.98\) as reported in Table~\ref{tab:model_params}-a.

\subsection{Modeling Insertion Cost}
\label{sec:insert_cost}

As discussed in Section~\ref{sec:background}, the insertion process runs a beam search to locate the new point's neighbors. This beam search performs the same distance computations that govern query cost, and it dominates the cost of insertion~\cite{HNSW, DiskANN}. We therefore measure the total number of distance computations during index construction, denoted $\mathrm{dc}_{\text{build}}$, and study how the average insertion cost $\mathrm{dc}_{\text{build}}/N$ changes with the dataset size and the construction parameters.

\Paragraph{Parameters} Since insertion runs the same beam search as a query, the insertion cost model shares the structure of the query cost models. The key difference is that insertion has no recall target, so recall demand drops out of the model. The parameters are the dataset size~$N$ and the two construction parameters, the graph out-degree $M$ and the candidate list size $ef_{\text{construction}}$.

\Paragraph{Expected Effects} A larger~$N$ means each new point searches a larger graph, so insertion cost should rise with the dataset, and since query cost rises with~$N$ as a Sublinear Power Law, we expect insertion cost to do the same; we empirically verify this power law. A larger $ef_{\text{construction}}$ makes each insertion explore more of the graph before identifying the new node's neighbors, so, just as query cost scales as $ef_{\text{search}}^{\,\beta}$ (Equation~\ref{eq:beta}), insertion cost should grow as a Sublinear Power Law in $ef_{\text{construction}}$. A higher $M$ forces the beam search to scan more candidates per step, but may also help it converge in fewer steps; the net effect is settled empirically.

\Paragraph{Modeling the Scaling} Figure~\ref{fig:cost_models}-b isolates each effect through partial regressions, as in Section~\ref{sec:query-cost-model}. At fixed $M$, the average insertion cost is log-linear in both $ef_{\text{construction}}$ and~$N$, as anticipated, verifying the power-law scaling on both counts. Raising $M$ increases the power-law exponent of~$N$, so a denser graph makes build cost grow faster with dataset size; $M$ therefore belongs inside the exponent of~$N$. We thus model
\begin{equation}
    \label{eq:build_model}
    \mathrm{dc}_{\text{build}}
    =
    e^{\kappa}\,
    (ef_{\text{construction}})^{\beta'}\,
    N^{\,c'},
    \qquad
    c' = 1 + c'_0 + \gamma'_2 \ln M.
\end{equation}
Equivalently,
\[
\ln\!\left(\frac{\mathrm{dc}_{\text{build}}}{N}\right)
=
\kappa
+ \beta' \ln(ef_{\text{construction}})
+ c'_0 \ln N
+ \gamma'_2 \ln(M)\ln(N).
\]
The primed symbols mark build-phase analogues of the query model's coefficients of Section~\ref{sec:query-cost-model}. The build exponent $c'$ explains how construction cost scales with the dataset size~$N$, the build-phase analogue of the query exponent~$c$. Within $c'$, the leading~$1$ is the trivial cost of performing~$N$ insertions, $c'_0$ captures how per-insertion search cost grows with index size, and $\gamma'_2$ captures how a larger $M$ steepens that growth. The intercept $\kappa$ depends on the data geometry and the graph construction algorithm, and $\beta'$ is the analogue of $\beta$ in Equation~\ref{eq:beta}, the exponent of the candidate list size.

\Paragraph{Goodness of Fit} Table~\ref{tab:model_params}-b lists the fitted coefficients for HNSW and Vamana. The fits are tight ($R^2 \geq 92\%$), and the slopes we recover line up with the partial regressions in Figure~\ref{fig:cost_models}-b. Except for the synthetic Rand64 dataset, $\beta'$ sits in a fairly narrow band, $0.80$--$0.92$, running $0.07$--$0.16$ above $\beta$. In other words, doubling $ef_{\text{construction}}$ raises the average insertion cost by about $1.8\times$. Insertion is thus more sensitive to its candidate list size than query cost is to $ef_{\text{search}}$ (Equation~\ref{eq:beta}). This gap comes from the extra distance computations insertion has to do beyond plain beam search, mostly reverse-link rewiring. This rewiring is triggered whenever a candidate neighbor has already hit its degree budget. Rand64 breaks this pattern: insertion is less sensitive to candidate list size than search, with $\beta'$ ($0.68$--$0.76$) staying below $\beta$ ($0.81$--$0.83$).

\section{Conclusion}
The literature has long held that, at fixed recall, the query cost of graph-based indexes scales poly-logarithmically. Our experiments and theory confine this claim to scales beyond a transition point. Below that point, cost follows a Sublinear Power Law, growing in step with $N^{c}$, and most datasets in practice stay in this regime up to their full size. The power law appears across eight datasets, recall targets from 90\% to 99\%, hard and easy queries, and many configurations up to billion scale. Cost bends toward poly-logarithmic only once the dataset densely samples its underlying distribution. Our unifying theory explains both regimes through the intrinsic dimensionality the dataset exhibits at its current size, reconciling our measurements with the logarithmic claims as two regimes of the same scaling law. The resulting cost models let practitioners set index parameters in a principled way and sustain a recall target as the dataset grows.

\bibliographystyle{ACM-Reference-Format}
\bibliography{sample}

\end{document}